%% file: cchpp.tex
\documentclass[a4paper, english, cleveref]{lipics-v2021}
\usepackage{booktabs}
\usepackage{KITcolors}
\usepackage{notation}
\usepackage{algorithm}
\usepackage{tikz}
\usetikzlibrary{calc,positioning,decorations,decorations.pathmorphing}
\usepackage{colortbl}

\usepackage{fnpct}
\setfnpct{after-punct-space=-0.15em}

\usetikzlibrary{backgrounds}

\title{Customizable Contraction Hierarchies -- A Survey}

\author{Thomas Bläsius}{Karlsruhe Institute of Technology, Germany}{}{}{}
\author{Valentin Buchhold}{Karlsruhe Institute of Technology, Germany}{}{}{}
\author{Dorothea Wagner}{Karlsruhe Institute of Technology, Germany}{}{}{}
\author{Tim Zeitz}{Karlsruhe Institute of Technology, Germany}{}{}{}
\author{Michael Zündorf}{Karlsruhe Institute of Technology, Germany}{}{}{}

\authorrunning{T. Bläsius, V. Buchhold, D. Wagner, T. Zeitz, and M. Zündorf}
\Copyright{Thomas Bläsius, Valentin Buchhold, Dorothea Wagner, Tim Zeitz, and Michael Zündorf}


\acknowledgements{This work was supported by funding from the pilot program Core-Informatics of the Helmholtz Association (HGF)}%?
\nolinenumbers

\EventEditors{John Q. Open and Joan R. Access}
\EventNoEds{2}
\EventLongTitle{42nd Conference on Very Important Topics (CVIT 2016)}
\EventShortTitle{CVIT 2016}
\EventAcronym{CVIT}
\EventYear{2016}
\EventDate{December 24--27, 2016}
\EventLocation{Little Whinging, United Kingdom}
\EventLogo{}
\SeriesVolume{42}
\ArticleNo{23}
\begin{document}

\maketitle

\begin{abstract}
This work establishes the technical fundamentals of a well-tuned Customizable Contraction Hierarchies (CCH) implementation that is simple and elegant.
We give a detailed overview of the state of the art of CCH, review recent advances on CCH and show how to combine them.
Additionally, we propose further refinements that improve the performance of CCH.
An extensive evaluation confirms that a CCH framework is not only comprehensive in supported features but also competitive in performance to both Contraction Hierarchies (CH) and Customizable Route Planning (CRP).
\end{abstract}

\section{Introduction}
A routing framework for modern map and navigation applications requires far more than an effective speedup technique for the quick computation of shortest paths~\cite{DBLP:journals/transci/DellingGPW17}.
While the shortest path computation should not be the bottleneck of the routing engine, there are many additional requirements.
A practical algorithm must handle not only car travel times but also other length functions (metrics) such as walking, biking, or user-defined preferences, for example, to avoid highways.
Further, \emph{dynamic} routing considering real-time traffic is critical.
Therefore, updates to the length function must be fast enough for frequent adjustments.
Integrating turn costs and restrictions is also an essential requirement for high-quality routing.
Finally, point-of-interest queries are also an important feature in such applications.

The first speedup technique that supports this entirety of requirements is \emph{Customizable Route Planning} (CRP) sometimes also referred to as Multi Level Dijkstra (MLD)~\cite{DBLP:journals/transci/DellingGPW17,DBLP:journals/tkde/DellingW15}.
This technique is also the first to take advantage of a three-phase approach where the classical preprocessing phase is split into an \emph{metric-independent phase} and a \emph{customization phase}.
The motivation behind this is that topological updates happen much less frequently than metric updates and thus only one part of the preprocessing has to be redone frequently.

Inspired by CRP, \emph{Contraction Hierarchies} (CH) have also been extended to a three-phase setup~\cite{DBLP:journals/jea/DibbeltSW16}.
This variant is called Customizable Contraction Hierarchies~(CCH).
Classical CH uses heuristic ad-hoc methods to build ``good enough'' axillary data in reasonable preprocessing running times.
These methods only work well for well-behaved weight functions with a strong hierarchy.
In contrast, CCH draws on nested dissections and algorithms from chordal graph theory for the augmented graph construction.
This yields surprisingly elegant and simple algorithms that also have stronger theoretical guarantees.
CCH queries are roughly as fast as CH queries, i.e., around an order of magnitude faster than CRP.
On the downside, the CCH customization is slower than the CRP customization thus far.
Also, the CCH proposed by Dibbelt et al.\ does not support all the features mentioned above.
However, CCH is an area of active research and in the past few years, there have been several research articles~\cite{DBLP:journals/jea/Buchhold0W19,DBLP:conf/wea/BuchholdW21,DBLP:conf/atmos/BuchholdWZZ20,DBLP:journals/algorithms/GottesburenHUW19} improving CCH and pushing it closer to feature parity with CRP.

\subparagraph{Contribution and Overview.}

In this work, we give a detailed overview of the state of the art of CCHs.
We review recent advances, show how to combine them, propose additional refinements, and provide a state-of-the-art reference implementation.
While this involves introducing, explaining, and discussing various concepts and details, we also strive to provide an accessible introduction to CCHs that is understandable without prior knowledge in route planning.
To achieve this, we start with a description on a high level of abstraction in Section~\ref{sec:cchpp:cch}.
This description is complete in the sense that it fully explains how and why CCH works, while leaving some freedom of how to exactly implement it.
In Section~\ref{sec:impl-deta} we discuss in detail the different ways of implementing the various aspects of a CCH, with a focus on the state of the art.
This is also the section where we propose several improvements to what has been done in the literature before.
This includes a simpler and more efficient algorithm for the metric-independent preprocessing in Section~\ref{sec:cchpp:metric_independent} and a better parallelization of the customization in Section~\ref{sec:cchpp:customization} (basic as well as perfect customization).
In Section~\ref{sec:cchpp:ext_queries}, we discuss the extended scenarios of point-of-interest queries as well as turn costs.
Moreover, we provide an adaptation of the Lazy RPHAST algorithm for one-to-many queries from CHs to CCHs, which in particular accelerates point-of-interest queries.
Finally, in Section~\ref{sec:cchpp:eval}, we evaluate our reference implementation.
Besides showing that our improvements are actually improvements, we also make a comparison to related techniques that are not the subject of this paper.

\section{Customizable Contraction Hierarchies}\label{sec:cchpp:cch}

At its core, the CCH framework has a preprocessing phase that augments the input graph with additional edges, which then speed-up subsequent shortest-path queries.
In the following, we give a concise introduction into how and why this works.
We assume that the input consists of an undirected graph $\aGraph = (\vertices, \edges)$ together with a \emph{weight function} $\len \colon \vertices \times \vertices \to \nonNegativeReals \cup \{\infty\}$ with $\len(u, v) = \infty$ for $\{u, v\} \notin E$.
Note that the graph itself is undirected while the weight function is directed.
Thus, for an edge $\{\uVertex, \vVertex\} \in E$, the weights $\len(\uVertex, \vVertex)$ and $\len(\vVertex, \uVertex)$ can be different and $\len(\uVertex, \vVertex)$ represents the cost to traverse the edge $\{\uVertex, \vVertex\}$ from $\uVertex$ to $\vVertex$.
This definition allows us to work with the graph's topology independent of directions but still allows asymmetric costs as well as one-way streets (using $\infty$ as weight for one direction).
The length of a path $\langle \vVertex_0, \dots, \vVertex_k\rangle$ is $\sum_{i = 1}^k \len(\vVertex_{i - 1}, \vVertex_i)$ and the distance from one vertex to another is the length of the shortest path between them.

The CCH preprocessing starts by assigning every vertex a unique \emph{rank} and constructing the \emph{augmented graph} $\augGraph = (\vertices, \augEdges)$ by contracting the vertices in \emph{rank order}, i.e., by increasing ranks.
For vertices $\uVertex, \vVertex \in \vertices$, we also write $\uVertex \prec \vVertex$ if the rank of $\uVertex$ is smaller than the rank of $\vVertex$.
To \emph{contract} $\vVertex \in \vertices$, complete the neighbors of $\vVertex$ with higher rank to a clique; see Figure~\ref{fig:cchpp:completion} for an example\footnote{We note that \augGraph is chordal and that the rank order is a perfect elimination scheme.}.
The edges created this way are called \emph{shortcuts}.
The idea is to perform shortest-path queries on \augGraph instead of \aGraph.
To take the new shortcut edges into account, we also define a modified weight function $\augLen \colon \vertices\times\vertices \to \nonNegativeReals \cup \{\infty\}$ for \augGraph.
The following observation provides the core insight that makes CCH work and shows us how \augLen should be defined.
We note that the discussion below the observation makes the somewhat technical observation easier to digest.

\begin{figure}[t]
	\centering%
	\input{Graphics/Completion.tex}%
	\caption{Visualization of the metric-independent preprocessing.
		Left: A graph \aGraph with seven vertices and an ranks implied by the $y$-coordinate of the vertices after $\uVertex_1$ and $\uVertex_2$ were processed.
		Right: The same graph after \vVertex was also processed.
		Note that only the neighbors of \vVertex that were not yet processed get completed into a clique.
	}\label{fig:cchpp:completion}
\end{figure}
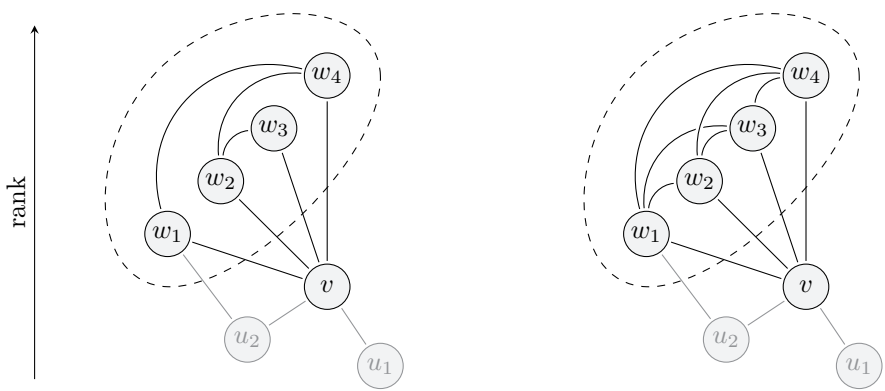

\begin{observation}
  \label{obs:core-observation}
  Let $\augLen$ be the weight function of $\augGraph$ such that for every edge $\{\vVertex, \wVertex\} \in \augEdges$,
  \begin{enumerate}
  \item \label{item:augmented-len-upper-bound}$\augLen(\vVertex, \wVertex)$ is at least the distance from $\vVertex$ to $\wVertex$ in $\aGraph$,
  \item \label{item:non-shortcut-len-lower-bound} the weights $\augLen$ are not larger than the original weights $\len$, i.e., $\augLen(\vVertex, \wVertex) \le \len(\vVertex, \wVertex)$, and
  \item \label{item:augmented-len-lower-bound}if $\{\uVertex, \vVertex\}, \{\uVertex, \wVertex\} \in \augEdges$ with $\uVertex \prec \vVertex$ and $\uVertex \prec \wVertex$, then $\augLen(\vVertex, \wVertex) \le \augLen(\vVertex, \uVertex) + \augLen(\uVertex, \wVertex)$.
  \end{enumerate}
  Then for every $\aSource, \aTarget \in \vertices$, and every shortest \aSource-\aTarget-path \aPath in \augGraph, one can obtain a shortest \aSource-\aTarget-path in \aGraph of the same length by replacing edges in \aPath with paths in \aGraph.
  Moreover, for every $\aSource, \aTarget \in \vertices$, there exists a shortest \aSource-\aTarget-path $\langle \vVertex_0, \dots, \vVertex_k \rangle$ with $\vVertex_0 = \aSource$ and $\vVertex_k = \aTarget$ in \augGraph such that $\vVertex_0 \prec \dots \prec \vVertex_i$ and $\vVertex_i \succ \dots \succ \vVertex_k$.
\end{observation}

Property~\ref{item:augmented-len-upper-bound} states that the inserted shortcuts are not shorter than the actual distance between the two endpoints, i.e., the distances in \augGraph are not shorter than in \aGraph.
Property~\ref{item:non-shortcut-len-lower-bound} makes sure that the non-shortcut edges in \augGraph are not too long, i.e., the distances in \augGraph are not longer than in \aGraph.
These two properties already imply the first statement in the observation, which asserts that shortest paths in \augGraph can be \emph{unpacked} to obtain shortest paths in \aGraph.
Unpacking here means that edges in \augGraph (in particular the shortcuts) are replaced by the paths they represent in \aGraph.

The second statement of the observation allows us to restrict the search for shortest paths in \augGraph to so-called \emph{up-down paths}, i.e., paths on which the rank first only goes up and then only down.
For this to hold, we need Property~\ref{item:augmented-len-lower-bound}, which is also called the \emph{lower triangle inequality}.
The name refers to the fact that $(\uVertex, \vVertex, \wVertex)$ is a \emph{lower triangle} of the edge $\{\vVertex, \wVertex\}$ in the sense that \uVertex has a lower rank than both endpoints of the edge; see Figure~\ref{fig:cchpp:triangle}.
Property~\ref{item:augmented-len-lower-bound} allows us to shortcut paths that are not up-down paths, i.e., paths that go at some point down and then up again.
With the observation that we only need to look for up-down paths, the query for a pair $\aSource, \aTarget \in \vertices$ can be simplified to search only upwards from \aSource and from \aTarget.
Motivated by this, we define the \emph{upward graph} $\gchu = (\vertices, \echu)$ to be the DAG obtained from $\augGraph$ by orienting each edge from low to high rank.
In the same way define the \emph{downward graph} $\gchd = (\vertices, \echd)$ by orienting all edges from high to low rank.
We denote a directed edge from $u$ to $v$ with $\uv$.

\begin{figure}[t]
	\centering%
	\input{Graphics/Triangle.tex}%
	\caption{Left: A triangle $(\uVertex, \vVertex, \wVertex)$ in \augGraph.
		We say that this is a lower triangle of the edge $\{\vVertex, \wVertex\}$ since \uVertex has lower rank than both \vVertex and \wVertex.
		In the same sense we say that it is an intermediate triangle of $\{\uVertex, \wVertex\}$ and an upper triangle of $\{\vVertex, \uVertex\}$.
		Right: visualization of the two inequalities Property~\ref{item:augmented-len-lower-bound} gives us for a single triangle, i.e., $\augLen(\vVertex, \wVertex) \le \augLen(\vVertex, \uVertex) + \augLen(\uVertex, \wVertex)$ and $\augLen(\wVertex, \vVertex) \le \augLen(\wVertex, \uVertex) + \augLen(\uVertex, \vVertex)$.
	}\label{fig:cchpp:triangle}
\end{figure}
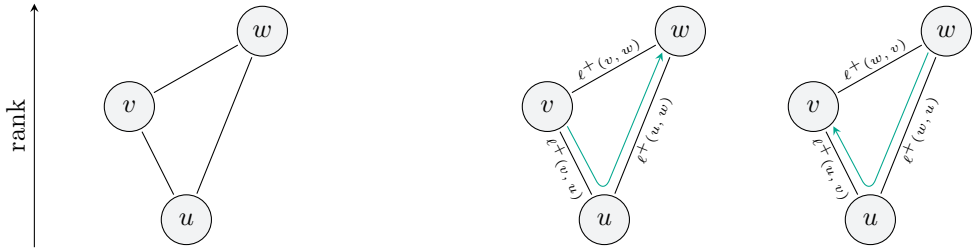

\subparagraph{The Three Phases of CCH.}
With the above considerations, we can now describe the three phases of the CCH framework in just a few sentences.
In the \emph{metric-independent preprocessing} one builds an augmented graph \augGraph based on a suitable rank order.
In the \emph{customization}, one computes the weight function \augLen of the augmented graph \augGraph.
In the \emph{query} phase, it then suffices to perform two upward searches in \gchu to find a shortest path.

A more detailed description follows in the next sections but for didactic reasons we will describe the three phases in reversed order starting with the query phase.
We note that we omit some implementation details for now and also skip extended scenarios, which we consider in the subsequent chapters.

\subsection{Queries}
\label{sec:queries}
In this phase we are given a source vertex \aSource and a target vertex \aTarget and we need to find a shortest up-down path between \aSource and \aTarget in \augGraph.
The classical algorithm to solve this is a bidirectional Dijkstra doing a \emph{forward search} from $s$ and a \emph{backward search} from $t$.
More specifically, the forward search from $s$ finds the up-part of the up-down path by running Dijkstra's algorithms on the upward graph $\gchu$ using $\augLen(\uVertex, \vVertex)$ as length for an edge $\uv$.
The backward search from $\aTarget$ finds the down-part of the up-down path.
It also runs on the upward graph $\gchu$ but uses $\augLen(\vVertex, \uVertex)$ as length for $\uv$ to reflect the fact that we are searching upward but are interested in a path oriented downward towards $\aTarget$.

One can simplify and speed up the query further by using the observations that $\gchu$ is a DAG and that the upward neighborhood of each vertex forms a clique~\cite{DBLP:journals/jea/DibbeltSW16}.
In the forward search, Dijkstra's algorithm would usually \emph{settle}\footnote{To settle a vertex $\uVertex$ means that every outgoing edge $\uv$ is checked for weather the path to $\vVertex$ via $\uVertex$ yields a new upper bound on the shortest path to $\vVertex$.} the vertices reachable form $\aSource$ in $\gchu$ in increasing distance from $\aSource$, and a priority queue is used to find the next vertex to settle in each step.
As $\gchu$ is a DAG, settling the vertices reachable form $s$ in any topological order is still correct.
Thus, if $\vVertex_0 \prec \dots \prec \vVertex_k$ are the vertices reachable from $s$ in $\gchu$ ordered by rank, then it is correct to settle them in this order.
Moreover, we can settle the vertices in this order very efficiently by using the following observations. 
As the upward neighborhood of every vertex forms a clique, $\gchu$ contains the edge $\vVertex_{i - 1}\vVertex_{i}$ for every $i \in [1, k]$ and $\vVertex_{i}$ is the upward neighbor of $\vVertex_{i - 1}$ with lowest rank.
Thus, if every vertex has a pointer to its upward neighbor of lowest rank, we just need to follow these pointers to visit all vertices reachable from $s$ ordered by rank.
The union of all of these pointers define the so-called \emph{elimination tree}.\looseness-1

As this clearly also works from $t$, the query boils down to doing a forward and backward search from $s$ and $t$, respectively, where each search just settles the vertices on the paths in the elimination tree from $s$ and $t$ to the root, which does not require any additional dynamic data structure.

\subparagraph{A Side Note.}

We note that the formalization in the beginning of Section~\ref{sec:cchpp:cch} of perceiving the input graph $\aGraph$ as undirected but with a directed weight function may seem strange at first.
To explain this, assume for a moment, that the input graph $\aGraph$ is a directed graph as shown in Figure~\ref{fig:cchpp:directed} (left).
Then one could argue that the edges $\uv$ and $\uw$ should not lead to shortcuts, as one cannot get from $\vVertex$ to $\wVertex$ via $u$ or vice versa.
Thus, not inserting an edge between $v$ and $w$ in the augmentation would still yield the correct result when using Dijkstra's algorithm for the forward and backward search in the query.
When using the elimination tree, however, it is crucial that $w$ is an upward neighbor of $v$.
Figure~\ref{fig:cchpp:directed} (right) highlights the elimination tree.

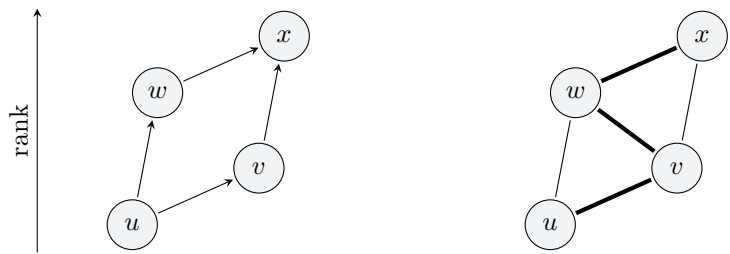
\begin{figure}[t]
	\centering%
	\input{Graphics/directedContraction.tex}%
	\caption{Left: The directed input graph \aGraph.
		Note that no additional edge would be created by the directed contraction of \uVertex since it has no incoming edges.
		Right: The outcome of the undirected contraction.
		We can see that an addition edge is created and that this edge is needed to form the elimination tree.
	}\label{fig:cchpp:directed}
\end{figure}

\subsection{Customization}\label{ssec:customization}

In the customization phase, we are given the input graph \aGraph with a weight function \len as well as the augmented graph \augGraph.
The goal is to compute the augmented weight function \augLen.
This is done as follows.
\begin{enumerate}
\item The augmented weight function \augLen is initialized with \len.
  This step is called \emph{respecting} and ensures Property~\ref{item:non-shortcut-len-lower-bound} of Observation~\ref{obs:core-observation}.
  After this, the inserted shortcuts have weight $\infty$, which is corrected in the next step.
\item The \emph{basic customization} step is the most critical part and modifies \augLen such that it satisfies all the properties in Observation~\ref{obs:core-observation}.
  For this, all lower triangles $(\uVertex,\vVertex,\wVertex)$ in \augGraph are processed in a bottom-up fashion, i.e., such that all triangles below \uv and \uw are processed earlier.
  For each such lower triangle, we decrease $\augLen(\vVertex, \wVertex)$ or $\augLen(\wVertex, \vVertex)$ if the corresponding paths via \uVertex are shorter.
  More precisely, as shown in Figure~\ref{fig:cchpp:triangle} we set
  \begin{align*}
    \augLen(\vVertex, \wVertex) &= \min\{\augLen(\vVertex, \wVertex), \augLen(\vVertex, \uVertex) + \augLen(\uVertex, \wVertex)\}\text{ and}\\
    \augLen(\wVertex, \vVertex) &= \min\{\augLen(\wVertex, \vVertex), \augLen(\wVertex, \uVertex) + \augLen(\uVertex, \vVertex)\}.
  \end{align*}
\end{enumerate}
After these two steps, \augLen satisfies the three properties of Observation~\ref{obs:core-observation}:
As \augLen is initialized with \len and only updated for edges for which we find shorter paths, $\augLen(\vVertex, \wVertex)$ is at least the distance from \vVertex to \wVertex in \aGraph (Property~\ref{item:augmented-len-upper-bound}).
As lengths are only decreased after the initialization, it clearly also holds that $\augLen(\vVertex, \wVertex) \le \len(\vVertex, \wVertex)$ (Property~\ref{item:non-shortcut-len-lower-bound}).
Moreover, after processing all lower triangles of $\{\vVertex, \wVertex\}$ in the basic customization step, the lower triangle inequality (Property~\ref{item:augmented-len-lower-bound}) must hold.
Due to processing the edges bottom-up, this remains true afterwards.

Although the above two steps are sufficient for correctness due to Observation~\ref{obs:core-observation}, there is some potential for improvement, as some edges might have a length greater than the shortest path between their endpoints.
Obviously, those edges will not be part of a shortest path and are therefore superfluous.
More precisely, any edge \uv for which there exists a shortest path using a vertex higher ranked than \uVertex or \vVertex is superfluous, as shown in~\cite{DBLP:journals/jea/DibbeltSW16}.
We can additionally identify those edges and remove them:
\begin{enumerate}
	\setcounter{enumi}{2}
	\item In the \emph{perfect customization}, triangles $(\uVertex,\vVertex,\wVertex)$ where $\uVertex \prec \vVertex$ and $\vVertex \prec \wVertex$ are processed again but in a top-down order.
	More precisely, we order them by \uVertex.
	
	Similar to lower triangles we will name the triangle $(\uVertex,\vVertex,\wVertex)$ an \emph{upper triangle} of the edge $\{\uVertex,\vVertex\}$ and an \emph{intermediate triangle} of the edge $\{\uVertex,\wVertex\}$.
	Again all triangles -- upper and intermediate -- each yield two inequalities that have to be considered, i.e., we set
	\begin{align*}
		\augLen(\uVertex, \vVertex) &= \min\{\augLen(\uVertex, \vVertex), \augLen(\uVertex, \wVertex) + \augLen(\wVertex, \vVertex)\}\text{,}\\
		\augLen(\vVertex, \uVertex) &= \min\{\augLen(\vVertex, \uVertex), \augLen(\vVertex, \wVertex) + \augLen(\wVertex, \uVertex)\}\text{,}\\
		\augLen(\uVertex, \wVertex) &= \min\{\augLen(\uVertex, \wVertex), \augLen(\uVertex, \vVertex) + \augLen(\vVertex, \wVertex)\}\text{, and}\\
		\augLen(\wVertex, \uVertex) &= \min\{\augLen(\wVertex, \uVertex), \augLen(\wVertex, \vVertex) + \augLen(\vVertex, \wVertex)\}.
	\end{align*}
	The first two inequalities consider the upper triangle of $\{\uVertex, \vVertex\}$ and the detour via \wVertex, whereas the last two inequalities consider the intermediate triangle of $\{\uVertex, \wVertex\}$ and the detour via \vVertex.
	
	Dibbelt et al. have shown that after all these triangles have been relaxed, all edges in the augmented graph have the weight of the shortest distance between their endpoints \cite{DBLP:journals/jea/DibbeltSW16}.
	\item Finally, in the \emph{construction} step, the minimal augmented graph \minAugGraph is constructed by removing any edge whose weight has changed during the perfect customization.
	However, if we use the same graph for forward and backward search we can only remove edges that are superfluous in both directions, i.e., if $\augLen(\wVertex, \vVertex)$ and $\augLen(\vVertex,\wVertex)$ have changed.
	Since most edges are only superfluous for either forward or backward search we build one graph with the relevant edges for each search.
\end{enumerate}
Note that the elimination tree defined in the previous section depends on \augGraph and not on \minAugGraph.
Therefore, the elimination tree should be calculated before any edge is removed, since after step four the elimination tree might no longer be a subgraph of the graph we are working with due to edge deletion.

\subsection{Metric-Independent Preprocessing}
\label{sec:metr-indep-prepr}

As mentioned above, we build the augmented graph \augGraph by contracting the vertices in rank order.
The degree of freedom we have to fill in is the choice of the rank order.
In theory, any order leads to correct queries.
However, good orders are critical for competitive performance.
On the one hand, we do not want to create too many shortcuts.
This would not only result in many triangles and therefore a slow customization but the augmented graph can also quickly become too large, even for the main memory of high-end server machines.
On the other hand, we want that the height of the elimination tree does not become too large as the query performance would suffer from this.

To motivate how the rank order is chosen, we first observe that a vertex separator can be used to split the problem of finding a good order in two independent subproblems.
Specifically, assume that $S \subseteq \vertices$ is a separator whose removal splits the graph \aGraph into at least two components $\aCell_1, \dots, \aCell_k$, which are also called \emph{cells}.
If we choose the rank order such that the separator vertices in $S$ come after all other vertices, then the order chosen for $\vertices \setminus S$ can be optimized independently for each cell $\aCell_i$.
More specifically, the resulting elimination tree will consist of a path for $S$ with a separate subtree attached to the lowest ranked vertex in $S$ for each of the cells $\aCell_i$.
The height of each of these subtrees only depends on the order chosen for the corresponding cell.
Also the set of inserted edges in $\aCell_i$ is independent of all other cells.
Applying this procedure recursively yields a so-called \emph{nested dissection order} \cite{DBLP:journals/tcs/BauerCRW16,g-ndrfe-73}.
Observe that the height of the elimination tree is the size of the separator $S$ plus the height of the highest subtree.
It thus makes sense to choose small balanced separators.
Additionally, choosing small separators also reduces the number of inserted shortcuts, as each separator is completed to a clique (at least if the separators are minimal with respect to inclusion).

One can use existing algorithms to actually compute a nested dissection order.
In Section~\ref{sec:cchpp:metric_independent}, we provide more details on which of the existing solutions to use and how to use it.
We note that there is also a strong connection to graph theory~\cite{DBLP:journals/tcs/BauerCRW16}, since the construction directly corresponds to the so-called \emph{elimination game} which has been extensively studied~\cite{DBLP:journals/dm/Heggernes06,parter1961use,rose1978algorithmic}.
There, the goal is to find a chordal supergraph by inserting few edges, which is done by iteratively removing vertices and completing their neighborhood to a clique.
The order of the vertices is then a perfect elimination scheme of the chordal supergraph.
The concept of the elimination tree, which we take advantage of during the query, has also been introduced in the context of the elimination game \cite{DBLP:journals/toms/Schreiber82,DBLP:journals/toms/Liu86}.

\section{Implementation Details}
\label{sec:impl-deta}

In this chapter we want to dive deeper into CCH.
We provide additional theoretical background with the goal of achieving a simple and efficient implementation.
Moreover, we discuss low-level optimizations that take the underlying hardware into account.

\subsection{Metric-Independent Preprocessing}\label{sec:cchpp:metric_independent}

Recall from Section~\ref{sec:metr-indep-prepr} that we use a nested dissection order to compute the augmented graph \augGraph, as it yields a small search space for the queries.
However, there are potentially many different rank orders that yield the same augmented graph and thus the same search space.
This leaves room for additional performance gains by choosing a different order that improves cache efficiency and jump prediction during the elimination tree query.
This section is structured as follows.
We first discuss which black-box algorithm to use for computing the initial nested dissection order.
Afterwards, we describe how to adapt the order to improve cache efficiency and jump prediction without changing the structure of \augGraph and the associated elimination tree.
For this, it is useful to observe that the nested dissection order itself encodes a separator decomposition that -- under natural assumptions -- is equal to the separator decomposition that the black-box algorithm computed.
Finally, we discuss how to efficiently construct the augmented graph \augGraph based on a given rank order.

\subparagraph{Initial Nested Dissection Order.}
Good nested dissection orders require the computation of small balanced separators, which is an \np-hard problem~\cite{DBLP:journals/tcs/GareyJS76}.
Thus, in general, one cannot expect to obtain small balanced separators efficiently.
Fortunately, road networks are easily separable due to natural features such as rivers or mountain ranges.
InertialFlow~\cite{DBLP:conf/wea/SchildS15} is an example of a simple yet surprisingly effective partitioning algorithm for road networks.
For this algorithm, vertices are projected on a geographic axis (the north-south axis, for example).
The first and last quarter of vertices are contracted, and a minimum cut between these contracted vertices is computed.
The vertices on one side of the cut edges make up a separator candidate.
Repeating this process for different axes and taking the smallest separator yields surprisingly good results.

More sophisticated algorithms for partitioning~\cite{DBLP:conf/alenex/SandersS12,DBLP:conf/wea/SandersS13} road networks~\cite{DBLP:conf/ipps/DellingGRW11,DBLP:journals/jea/HamannS18} have also been developed.
A comprehensive review is beyond the scope of this work.
See~\cite{DBLP:journals/algorithms/GottesburenHUW19} for a recent experimental comparison.
To the best of our knowledge, InertialFlowCutter~(IFC)~\cite{DBLP:journals/algorithms/GottesburenHUW19}, a combination of InertialFlow~\cite{DBLP:conf/wea/SchildS15} and FlowCutter~\cite{DBLP:journals/jea/HamannS18}, currently yields the best results.
We use IFC orders for all experiments in this work.

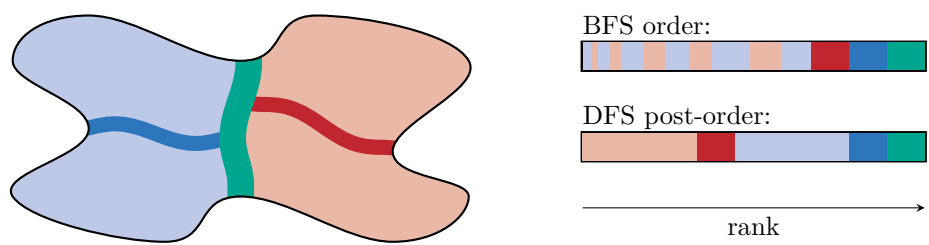
\begin{figure}
	\centering
	\centering
	\input{Graphics/SeperatorDecomposition.tex}
	\caption{
		Visualization of different rank orders derived from the same separator decomposition.
		Left: A graph with a top level separator colored in green that splits the graph into a red and blue cell.
		The separators of the two cells are also colored in dark red and dark blue respectively.		
		Right: A BFS order and DFS post-order visualized.		
		In the BFS order the top-level separator vertices appear last in the order.
		And all vertices of the same layer form a continuous range in the order.
		In the DFS post-order, again, the top-level separator vertices appear last in the order.
		However, in this order sub-trees form continuous ranges.
		Therefore, the separator vertices of the blue cell \emph{and} the entire blue cell appear after all vertices of the red cell.
	}\label{fig:cchpp:sep_decomp}
\end{figure}

\subparagraph{Improved Rank Order.}
To think about cache efficiency in the query, recall that the query essentially consists of walking up paths to the root in the elimination tree.
This becomes more cache efficient if there are many subpaths with consecutive vertex IDs.
To achieve this, we change vertex IDs to match the rank order, i.e., the vertex that is contracted first gets the lowest ID.
Recall that vertices from the same separator are contracted consecutively and form a path in the elimination tree.
Thus, this reordering of the IDs already yields more paths with consecutive IDs.
We can further increase the cache efficiency by choosing the rank order to be a DFS post-order of the elimination tree.
To see that this is better than, e.g., the BFS post-order that is provided by InertialFlowCutter, consider Figure~\ref{fig:cchpp:sep_decomp}.

To compute the improved order, we first compute the augmented graph \augGraph based on the initial nested dissection order provided by the black-box algorithm and thereby obtain the elimination tree.
From this, we obtain the improved rank order.
We note that this detour has no algorithmic reasons but are required because one wants to be able to use the nested dissection algorithms as black-box

\subparagraph{Reconstructing Separator Decompositions.}
\label{sec:reconstruct_sep_tree}

As mentioned before, the black-box algorithms only yields an order and not the separator decomposition.
From this order, we can clearly do the contraction and thereby compute \augGraph and the elimination tree.
Although this is sufficient for the core algorithm, it is sometimes useful to actually know the separator decomposition, not only to improve the rank order but also for nearest neighbor search or parallelization.
We have mentioned before that there is a strong connection between the elimination tree and the separator decomposition.
The following lemma formalizes this connection.

\begin{lemma}
  \label{lem:separator-hierarchy-from-elim-tree}
	Let $T$ be the elimination tree obtained from a nested dissection order and let \vVertex be the highest-ranked vertex with more than one child in $T$.
	Then the set $S$ of vertices on the path from \vVertex to the root is a separator in \aGraph .
\end{lemma}
\begin{proof}  
  Let \uVertex and \wVertex be two children of \vVertex and assume without loss of generality that $\uVertex \prec \wVertex$.
  Further, let $T_{\uVertex}$ and $T_{\wVertex}$ denote the subtrees rooted at \uVertex respectively \wVertex in $T$.  
  Assume for contradiction that $\augGraph$ has an edge from a vertex in $T_{\uVertex}$ to a vertex in $T_{\wVertex}$.
  This is equivalent to the existence of a path \aPath from \uVertex to \wVertex in \augGraph that avoids $S$.
  Let $\uVertex'$ be the lowest ranked vertex on \aPath and assume without loss of generality that \aPath is chosen such that $\uVertex'$ has the highest rank among all choices of \aPath.

  We first argue that $\uVertex'$ has to be the end of \aPath, which implies that $\uVertex' = \uVertex$ as $\uVertex \prec \wVertex$.
  Otherwise, $\uVertex'$ is adjacent to two vertices in \aPath.
  However, those two vertices have to be connected by an edge after the contraction of $\uVertex'$.
  Therefore, there is a path connecting \uVertex and \wVertex that is equivalent to $P$ except that it skips the lowest ranked vertex $\uVertex'$, contradicting our choice of \aPath.
  
  Secondly, we argue that the existence of \aPath implies that \uVertex is not a child of \vVertex.
  Let $\uVertex''$ be the neighbor of \uVertex in \aPath.
  As $\uVertex$ is the lowest ranked vertex on \aPath, $\uVertex''$ has a higher rank.
  Further, $\uVertex''$ must have a lower rank than \vVertex since
  \aPath avoids $S$ and therefore all vertices with higher or equal rank.
  However, this implies that $\uVertex''$ should be the parent of \uVertex and not \vVertex.
  
  In conclusion this implies that no path avoiding $S$ exists in \augGraph.
  Therefore, no edge connecting $T_{\uVertex}$ and $T_{\wVertex}$ exists in \augGraph and hence, no such edge exists in \aGraph.
\end{proof}

\noindent
We can also conclude that $S$ is a subset of the top-level separator of the original nested dissection order.
If we assume that the black-box algorithm found an inclusion minimal separator, then $S$ is the original top-level separator.
Since $S$ is a separator, removing it from the original graph leaves two or more disconnected cells.
Each cell contains the vertices of a subtree of $T$ where the root is a child of \vVertex.
Recursively applying this idea to each subtree lets us obtain the separator decomposition for the whole graph.

\subparagraph{Contraction.}\label{sec:cchpp:contraction}
Once an order has been obtained, the remaining part of the metric-independent preprocessing is to compute the topology of the augmented graph $\augGraph =(\vertices, \augEdges)$.
For cache efficiency in this and all following phases, it is crucial to first permute the vertex IDs such that IDs equal the ranks.
Moreover, since the graph is undirected, we can increase efficiency by only storing each edge at its lower-ranked endpoint, i.e., maintain \gchu, which is also sufficient for all further steps.

The naive way to construct \gchu is to iterate over all vertices by increasing rank and to ensure that an edge exists between any pair of upward neighbors, inserting shortcuts where the edges are missing.
Because this approach is relatively expensive in terms of running time, Dibbelt, Strasser, and Wagner~\cite{DBLP:journals/jea/DibbeltSW16} propose a faster algorithm based on the \emph{quotient graph}~\cite{gl-aqgms-78}.
However, this algorithm is quite complex.

Here, we describe an algorithm that is simpler and faster.
This algorithm has previously only been described in a Bachelor's thesis~\cite{z-cchtc-19}.
It is heavily based on the linear-time lex-BFS algorithm for chordal graph recognition~\cite{DBLP:journals/tcs/HabibMPV00}.
The key observation of this algorithm is that the naive algorithm considers the same edges multiple times in order to ensure that the upwards neighborhood for each vertex is a clique.
This can be prevented using the following observation.

\begin{observation}
  Consider a vertex \uVertex and its lowest upward neighbor \vVertex, i.e., its parent in the elimination tree.
  Then, in the final graph \augGraph, the upward neighborhood of \vVertex is a superset of the upward neighborhood of \uVertex (up to \vVertex itself).
\end{observation}

With this in mind we can delay parts of the completion of \uVertex's upward neighborhood to later.
Assuming that the algorithm later completes the upward neighborhood of \vVertex to a clique (this is essentially the induction hypothesis), then the upward neighborhood of \uVertex is also completed to a clique if we make sure that every upward neighbor of \uVertex (except \vVertex) is also an upward neighbor of \vVertex.
Thus, when processing \uVertex, it suffices to set $\mathtt{N}[\vVertex] \gets \mathtt{N}[\vVertex] \cup \mathtt{N}[\uVertex] \setminus \vVertex$.

There are some degrees of freedom for how to implement this procedure.
We propose to use the variant shown in Algorithm~\ref{algo:elimination}.
In contrast to the high level description before, we do not union neighborhoods but simply concatenates them.
This temporarily creates duplicate entries in the neighborhoods that need to be removed later.
However, Zündorf showed that this does not change the theoretical worst-case running time~\cite{z-cchtc-19}.
In addition, we choose to not only compute \augGraph, but also store the neighborhoods sorted by rank.
This is not only beneficial for further computations, it also allows to remove duplicates efficiently with a single linear sweep.
Note that we delay both steps -- sorting and deduplication -- until a vertex is processed.
This is not asymptotically optimal but simpler to implement and equally fast in practice.

\begin{algorithm}[tb]
	\caption{Contraction algorithm.}
	\label{algo:elimination}
	\KwData{\gchu: The input graph with edges directed to higher ranked vertices.}
	\SetKw{Continue}{continue}
	\SetKwFunction{sorted}{Sorted}
	\SetKwFunction{unique}{Unique}
	\SetKwFunction{contraction}{Contraction}
	
	\BlankLine
	\Fn{\contraction}{
		\For{all vertices $\uVertex \in \vertices$ in rank order}{
			\If{$\mathtt{N}[\uVertex]$ not empty}{
				$\mathtt{N}[\uVertex] \gets \sorted(\mathtt{N}[\uVertex])$ \tcp{sort edges \uv by increasing rank of \vVertex}
				$\mathtt{N}[\uVertex] \gets \unique(\mathtt{N}[\uVertex])$\;
				$\uv \gets \mathtt{N}[\uVertex][0]$\;\label{algo:elimination:parent}
				$\mathtt{N}[\vVertex] \gets \mathtt{N}[\vVertex] + \mathtt{N}[\uVertex][1\dots]$\;
			}
		}
	}
\end{algorithm}

\subparagraph{Elimination Tree.}
We represent the elimination tree as an array \elimtree of length \numVertices where position ID containing the parent for the vertex with that ID.
As root vertices have no parent, we represent this with some sentinel value, which we denote as $\bot$.
Note that the parent is already computed during the contraction in Line~\ref{algo:elimination:parent} of Algorithm~\ref{algo:elimination} and only needs to be stored.

\subsection{Customization}\label{sec:cchpp:customization}

We now want to explain the customization on a deeper level.
First, recall that for this phase we get access to the edge lengths \len and need to compute the edge lengths for the augmented graph.
As mentioned in Section~\ref{sec:cchpp:contraction}, we only maintain the upwards edges of the augmented graph \augGraph.
This is  sufficient from a topological point of view, but we still need to handle different edge weights for the forward and backward search.
For this, we store the weights in two arrays $\lchu$ and $\lchd$ accessible by the corresponding edge IDs from \echu, i.e., for $\uv \in \echu$, $\lchu[\uv] = \augLen(\uVertex, \vVertex)$ and $\lchd[\uv] = \augLen(\vVertex, \uVertex)$.

\subsubsection{Basic Customization}
\label{sec:cchpp:basic}

Recall from Subsection~\ref{ssec:customization} that we need to enforce the lower triangle inequality for all triangles in \augGraph, i.e., when processing the triangle $(\uVertex,\vVertex,\wVertex)$ with $\uVertex \prec \vVertex \prec \wVertex$, the length of \vw is potentially updated based on the lengths of \uv and \uw.
As we do not want to process $(\uVertex,\vVertex,\wVertex)$ multiple times, the triangles are enumerated in an order that processes every lower triangle of \uv and \uw before $(\uVertex,\vVertex,\wVertex)$.
We call such an order a \emph{bottom-up order}.
In the following, we discuss the advantages and disadvantages of different orders that satisfy this property as well as how to actually enumerate the triangles.

Two bottom-up orders have been proposed in the literature~\cite{DBLP:journals/jea/DibbeltSW16, DBLP:journals/jea/Buchhold0W19}.
Both variants entail iterating over the edges in a lexicographic order, i.e., sorted by the rank of the lower and upper endpoint.
This can be easily achieved by iterating the vertices in rank order and then, for each vertex \uVertex, iterating the outgoing neighbors of \uVertex in rank order.
For each edge \uv, the original CCH publication by Dibbelt, Strasser and Wagner then enumerates all lower triangles of \uv~\cite{DBLP:journals/jea/DibbeltSW16}, which clearly yields a bottom-up order.
Buchhold, Sanders and Wagner propose to instead enumerate all upper triangles of \uv~\cite{DBLP:journals/jea/Buchhold0W19}.
We note that in the latter case, the enumerated triangle $(\uVertex, \vVertex, \wVertex)$ is updated as lower triangle for \vw, which can be shown to indeed also yield a bottom-up order.

The first order by Dibbelt et al.~\cite{DBLP:journals/jea/DibbeltSW16} has the benefit that, for each edge \uv all lower triangles of \uv are consecutive.
Thus, all changes to the length of \uv happen consecutively, which simplifies parallelization.
However, since this approach enumerates the triangles below a given edge \uv, it is necessary to store the incoming edges, i.e., this requires to at least temporarily compute $\echd$.
The second order by Buchhold et al.~\cite{DBLP:journals/jea/Buchhold0W19} enumerates upper triangles of \uv, which does not require $\echd$ and additionally turns out to be more efficient.
However, since the edges that get updated are more widely distributed, the parallelization is more difficult.

In the following, we first discuss how to actually enumerate the triangles in these orders using a so-called coordinated linear sweep.
We also briefly discuss what makes the approach by Buchhold et al.\ more efficient.
Afterwards, we propose a new way to enumerate the triangles that we call batched triangle relaxation.
It is as efficient as the variant by Buchhold et al.\ while also being easy to parallelize.

\subparagraph{Coordinated Linear Sweep.}

Both Dibbelt et al.~\cite{DBLP:journals/jea/DibbeltSW16} and Buchhold et al.~\cite{DBLP:journals/jea/Buchhold0W19} utilize the fact that the intersection of two sorted lists can be found by iterating through both simultaneously.
To be more precise, for a fixed edge \uv they iterate the neighborhoods of \uVertex and \vVertex in sorted rank order.
Whenever a vertex is seen in both neighborhoods this completes a triangle.  
This process is also refered to as a \emph{coordinated linear sweep}.

The only difference here is that Dibbelt et al.\ consider the downwards neighborhood whereas Buchhold et al.\ use the upwards neighborhood.
However, this difference is critical for performance.
Buchold et al.\ observed that their approach of enumerating upper triangles is significantly faster due to the distribution of the vertex in-degrees and out-degrees~\cite{DBLP:journals/jea/Buchhold0W19}.

Additionally, enumerating upper triangles allows for a simpler implementation of the coordinated linear sweep using the fact that the augmented graph is chordal.
This essentially turns it into a linear sweep (getting rid of the ``coordinated'').
To explain this, let \uv be the edge for which we want to enumerate all up upper triangles.
As the upwards neighborhood of each vertex is a clique, the upwards neighborhood $\mathtt{N}[\uVertex]$ is a subset of $\mathtt{N}[\vVertex]$.
Therefore, each vertex in $\mathtt{N}[\uVertex]$ yields exactly one triangle all of which can be found with a single linear sweep through $\mathtt{N}[\vVertex]$ as shown in Algorithm~\ref{algo:basic_customization_buchhold}.

\begin{algorithm}[tb]
	\caption{Basic customization algorithm with linear sweep.}
	\label{algo:basic_customization_buchhold}
	\KwData{$\lchu[\uv]$: length of $\uv \in \echu$ in the augmented graph.}
	\KwData{$\lchd[\uv]$: length of $\vu \in \echd$ in the augmented graph.}
	\SetKwFunction{Basic}{BasicCustomization}
	
	\BlankLine
	\Fn{\Basic}{
		\For{all vertices $\uVertex \in \vertices$ in rank order}{
			\For{\uv in $\echu$}{
				$k \gets 0$\;
				\For{\uw in $\echu$ sorted by $\wVertex$ and $\vVertex \prec \wVertex$}{
					\While{$\mathtt{N}[\vVertex][k] \neq \vw$}{\label{algo:basic_customization_buchhold:sweep}
						$k \gets k+1$\;
					}
					\tcp{$(\uVertex,\vVertex,\wVertex)$ is a triangle with $\uVertex \prec \vVertex \prec \wVertex$}
					$\lchu[\mathtt{N}[\vVertex][k]] \gets \min(\lchu[\mathtt{N}[\vVertex][k]], \lchd[\uv] + \lchu[\uw])$ \tcp{forward}
					$\lchd[\mathtt{N}[\vVertex][k]] \gets \min(\lchd[\mathtt{N}[\vVertex][k]], \lchu[\uv] + \lchd[\uw])$ \tcp{backward}
				}
			}
		}
	}
\end{algorithm}

\subparagraph{Batched Triangle Relaxation.}

Here we propose a new way to enumerate triangles.
Its efficiency is on-par with Algorithm~\ref{algo:basic_customization_buchhold} while also being easy to parallelize.
It is based on two ideas.
The first idea is used since the eighties~\cite{DBLP:journals/siamcomp/ChibaN85} to enumerate triangles but has not been used for CCHs before.
The other is based on choosing yet another order of triangles.
The only downside of this new approach is that it also temporarily requires the downward edges \echd.
Thus, the Algorithm~\ref{algo:basic_customization_buchhold} remains the preferred option for the case where no parallelization is required.
If parallelization is required, the approach we describe in the following is the method of choice.
Our reference implementation uses parallelization and we already compute \echd during the metric-independent preprocessing.

The first core idea, ignoring edge directions for now, is the following~\cite{DBLP:journals/siamcomp/ChibaN85}.
To enumerate all triangles containing a specific vertex \uVertex, first mark all neighbors of \uVertex.
Then, enumerate all paths consisting of two edges starting at \uVertex, i.e., consider a neighbor \vVertex of \uVertex and a neighbor \wVertex of \vVertex.
Then $(\uVertex, \vVertex, \wVertex)$ forms a triangle if and only if \wVertex is a neighbor of \uVertex.
This can be checked in constant time as we marked all the neighbors of \uVertex.
After having enumerated all triangles containing \uVertex, we unmark the neighbors of \uVertex.
As we are not only interested in whether $(\uVertex, \vVertex, \wVertex)$ is a triangle but also need to know the edge ID of \uw in case the answer is yes, we use the corresponding edge IDs to mark the neighbors of \uVertex (instead of just setting a Boolean flag).

The second core idea is to choose an order that, for every vertex \uVertex, enumerates all triangles for which \uVertex is the central vertex consecutively.
We call these triangles the \emph{batch} of $u$.
Note that, for any edge $\uv \in \echu$, all lower triangles of \uv are in the same batch, namely in the batch of \uVertex.
Thus, the length of \uv is only updated in one batch and we can easily run the different batches in parallel.
Note that this is a weaker condition than actually enumerating all lower triangles of \uv consecutively.
This gives us the freedom to choose an order that has better performance due to in/out-degree effects, similar to what makes enumerating upper triangles faster than enumerating lower triangles.

A slightly different way to view this is the following.
Enumerating triangles essentially boils down to enumerating paths consisting of two edges and then checking whether the third edge is present.
Using the order by Dibbelt et al.~\cite{DBLP:journals/jea/DibbeltSW16} these paths are up-down paths, i.e., starting from a vertex, one goes up by iterating all outgoing edges and from the target vertex one goes down by iterating all incoming edges.
Due to how in and out degrees are distributed, up-down this is slower than the variant by Buchhold et al.~\cite{DBLP:journals/jea/Buchhold0W19} enumerating up-up paths.
To enumerate the triangles in batches, we cannot choose up-up paths.
However, we can choose down-up paths, which turns out to similarly efficient.

Using these core ideas yields Algorithm~\ref{algo:basic_customization}, which we call \emph{batched triangle relaxation}.
It iterates over all vertices \uVertex in rank order.
It then first marks all upward neighbors \vVertex with the corresponding edge ID \uv; see Line~\ref{algo:batched_triangle_basic:mark_neighbors}.
Once all triangles that have \uVertex as central vertex are enumerated, the upward neighbors of \uVertex are unmarked in Line~\ref{algo:batched_triangle_basic:cleanup}.
To enumerate these triangles, we iterate over all downward neighbors \wVertex of \uVertex and then over all upward neighbors \vVertex of \wVertex that are above \uVertex.
This can be achieved by iterating the upward neighbors of \wVertex in reverse rank order and stopping, once we reach \uVertex; see Line \ref{algo:batched_triangle_basic:break}.
If the edge \uv exists, which can be checked due marking, then the triple $(\wVertex, \uVertex, \vVertex)$ forms an upper triangle of \uv.

To show that our batched triangle relaxation is correct, we have to argue that the resulting triangle order is a bottom-up order.
For this, let $(\wVertex, \uVertex, \vVertex)$ be a triangle belonging to the batch of \uVertex with $\wVertex \prec \uVertex \prec \vVertex$.
We have to argue that all lower triangles of \wu and all lower triangles of \wv have been processed before $(\wVertex, \uVertex, \vVertex)$.
Note that lower triangles of \wu and lower triangles of \wv have \wVertex as central vertex.
Thus these triangles have been processed in the batch of \wVertex, which happens before processing the batch of \uVertex containing $(\wVertex, \uVertex, \vVertex)$.

\begin{algorithm}[tb]
	\caption{Basic customization algorithm with batched triangle relaxing.}
	\label{algo:basic_customization}
	\KwData{$\lchu[\uv]$: length of $\uv \in \echu$ in the augmented graph.}
	\KwData{$\lchd[\uv]$: length of $\vu \in \echd$ in the augmented graph.}
	\KwData{$\mathtt{ID}[\vVertex]$: edge ID of \uv where \uVertex is the current vertex in the outer loop, initially~\makebox[0pt][l]{$\bot$.}}
	\SetKw{Break}{break}
	\SetKwFunction{Basic}{BasicCustomization}
	
	\BlankLine
	\Fn{\Basic}{
		\For{all vertices $\uVertex \in \vertices$ in rank order}{
			\label{algo:batched_triangle_basic:main_loop}
			\For{all edges $\uv \in \echu$}{
				$\mathtt{ID}[\vVertex] \gets \uv$\;
                          \label{algo:batched_triangle_basic:mark_neighbors}
			}
			\For{all edges $\uw \in \echd$}{ \label{algo:batched_triangle_basic:echd}
				\For{all edges $\wv \in \echu$ ordered by descending rank of \vVertex}{
					\If{$\vVertex \preceq \uVertex$}{
						\Break\;
						\label{algo:batched_triangle_basic:break}
					}
					\tcp{$(\wVertex,\uVertex,\vVertex)$ is a triangle with $\wVertex \prec \uVertex \prec \vVertex$}
					\If{$\mathtt{ID}[\vVertex] \neq \bot$}{
						$\lchu[\mathtt{ID}[\vVertex]] \gets \min(\lchu[\mathtt{ID}[\vVertex]], \lchd[\uw] + \lchu[\wv])$ \tcp{forward}
						$\lchd[\mathtt{ID}[\vVertex]] \gets \min(\lchd[\mathtt{ID}[\vVertex]], \lchu[\uw] + \lchd[\wv])$ \tcp{backward}
					}
				}
			}
			\For{all edges $\uv \in \echu$}{
				\label{algo:batched_triangle_basic:cleanup}
				$\mathtt{ID}[\vVertex] \gets \bot$\;
			}
		}
	}
\end{algorithm}

\subparagraph{Unpacking data.}
So far we have shown the minimum required customization steps to get correct queries as outlined in Observation~\ref{obs:core-observation}.
However, most application not only need to compute the shortest distance between two vertices but also the corresponding path that connects them.
One important step of this is \emph{unpacking} edges from $\augEdges \setminus \edges$ and replace them with the corresponding edges from \edges.
The original CCH publication \cite{DBLP:journals/jea/DibbeltSW16} suggests to not maintain any additional data for this process and instead unpack edges at query time by enumerating lower triangles.
This minimized the customization time and memory requirements at the cost of slower unpacking and therefore slower query times.
If faster queries are desirable (at the cost of slower customization), we can store additional unpacking data along each edge to allow faster unpacking.
One options is to store the lower vertex of the shortest lower triangle.
However, the largest speed up for the query can be achieved by storing both edges forming the shortest triangle along each edge.
Note that in both cases this data has to be stored for the forward and backward search separately.
With this data each edge can be unpacked asymptotically optimal, i.e., in time proportional to output size.

\subsubsection{Perfect Customization}
\label{sec:cchpp:perfect}

Recall from Subsection~\ref{ssec:customization} that we can improve query times by identifying and removing edges in \augGraph that are not part of shortest paths for the given metric.
For this we need to enumerate all triangles again and relax each one as upper and intermediate triangle (and not as lower triangle as for the basic customization).
As before, choosing the right order of triangles makes it sufficient to consider each triangle only once.
This time, wee need to use a \emph{top-down} order.
More precisely, Dibbelt et al.\ showed that ordering the triangles by the rank of the lowest ranked vertex allows to to identify all edges that can be removed~\cite{DBLP:journals/jea/DibbeltSW16}.
We note that this is not obvious and uses the fact that the graph is chordal.
We also want to note that this notion of top-down is not inverse to the notion of bottom-up used for the basic customization in Section~\ref{sec:cchpp:basic}.

Algorithm~\ref{algo:perfect_customization} shows how to enumerate the triangles in a top-down order, using a (coordinated) linear sweep that is almost identical to the one in Algorithm~\ref{algo:basic_customization_buchhold} for the basic customization.
The only difference is that the vertices are processed in reverse rank order.
Recall that, for the basic customization, Algorithm~\ref{algo:basic_customization_buchhold} had the disadvantage that parallelization is difficult as the potential length updates are scattered throughout the order.
Fortunately, this is not an issue for the perfect customization.
For a vertex \uVertex, Algorithm~\ref{algo:perfect_customization} enumerates all triangles with \uVertex as bottom vertex and then potentially updates one of the two edges incident to \uVertex (and not the edge opposite to \uVertex as in the basic customization).
Thus, all potential changes to outgoing edges of \uVertex happen consecutively.

\begin{algorithm}[tb]
	\caption{Perfect customization algorithm with batched triangle relaxing.}
	\label{algo:perfect_customization}
	\KwData{$\lchu[\uv]$: length of $\uv \in \echu$ in the augmented graph.}
	\KwData{$\lchd[\uv]$: length of $\vu \in \echd$ in the augmented graph.}
	\SetKwFunction{Perfect}{PerfectCustomization}
	
	\BlankLine
	\Fn{\Perfect}{
		\For{all vertices $\uVertex \in \vertices$ in reversed rank order}{\label{algo:perfect_customization:order}
			\For{\uv in $\echu$}{
				$k \gets 0$\;
				\For{\uw in $\echu$ sorted by $\wVertex$ and $\vVertex \prec \wVertex$}{
					\While{$\mathtt{N}[\vVertex][k] \neq \vw$}{
						$k \gets k+1$\;
					}
					\label{algo:perfect_customization:relax}
					\tcp{$(\uVertex,\vVertex,\wVertex)$ is an upper triangle of \uv}
					$\lchu[\uv] \gets \min(\lchu[\uv], \lchu[\mathtt{N}[\vVertex][k]] + \lchd[\uw])$ \tcp{forward}
					$\lchd[\uv] \gets \min(\lchd[\uv], \lchd[\mathtt{N}[\vVertex][k]] + \lchu[\uw])$ \tcp{backward}
					
					\tcp{$(\uVertex,\vVertex,\wVertex)$ is an intermediate triangle of \uw}
					$\lchu[\uw] \gets \min(\lchu[\uw], \lchu[\uv] + \lchu[\mathtt{N}[\vVertex][k]])$ \tcp{forward}
					$\lchd[\uw] \gets \min(\lchd[\uw], \lchd[\uv] + \lchd[\mathtt{N}[\vVertex][k]])$ \tcp{backward}
				}
			}
		}
	}
\end{algorithm}

Recall that the perfect customization is not necessary to obtain correct shortest paths.
Its purpose is to identify edges that are superfluous and can thus be removed.
An edge is superfluous if its length changes during the perfect customization.
Thus, these edges have to be marked for deletion during the perfect customization and can then be removed afterwards.
Note that this also means that we do not need to update unpacking information when changing an edge length, as these edges will be deleted anyways.
For an efficient parallel implementation of marking superfluous edges for deletion, these markings must be individually addressable, i.e., they must stored in bytes not bits.
With a more compact representation, two threads might need to write to the same address leading to race conditions.

\subparagraph{Reduced Graph Construction.}

After the perfect customization we have to remove all superfluous edges from \augGraph.
To be more precise, we want to construct one graph for the forward search and one graph for the backward search each of which only contains edges that are not marked for that direction.
The original CCH publication~\cite{DBLP:journals/jea/DibbeltSW16} does not provide details on this process, probably because it appears trivial from an algorithmic point of view.
However, building the reduced augmented graph makes up a significant share of the total running time of the customization~\cite{DBLP:journals/jea/Buchhold0W19}.
Therefore, an efficient parallelization is essential.
In the following, we describe how this can be achieved.
We note that this part is somewhat technical since our approach is specific for our graph representation which we need to define first.

In our graph representation all edges are stored sequentially in an array.
The only requirement we need is that the edges are grouped and sorted by their tail, i.e., if $\uVertex\prec\vVertex$ all edges from \uVertex appear before any edge from \vVertex.
Thus, the neighborhood of each vertex forms a continuous range in the edge array.
To be able to efficiently determine these ranges for each vertex, we also store the ID of the first edge for each vertex in a second array.

With this graph representation in mind our algorithm needs to complete three tasks.
First it has to compute the new ranges of edge IDs for each vertex.
To be more precise it has to compute the new ID of the first edge for each vertex.
After that it has to move all edges we want to keep and their associated auxiliary data into a new array.
Note that we again want to store all edges sequentially without gaps.
In the last step, one has to update auxiliary data that depends on edge IDs (which have changed) like unpacking information.

To parallelize these steps we propose to split the graph into $\betaParameter \cdot \numCores$ consecutive chunks of vertices.
Here \numCores is the number of cores and \betaParameter a tuning parameter, which we set to $4$ by default.
Note that we want each chunk to contain roughly the same number of outgoing edges not vertices.
Therefore, we find the ID of the first vertex of chunk $i$ by taking the tail of the edge with ID $i \cdot \lceil \frac{\numEdges}{\betaParameter \numCores} \rceil$, i.e., we split based on the number of edges but round by neighborhoods.

To obtain the new edge IDs we first process all chunks in parallel to count the edges that will remain in that chunk.
Afterwards a sequential prefix sum over these sums yields the reduced edge ID range for each chunk.
Afterwards, we process all chunks again in parallel.
Starting from the first edge ID for the chunk the first edge ID for each vertex can be computed.
Additionally, we can already copy all edges and the associated data to their new destination.
The associated edge data consists at the very least of the head vertex of each edge and the weight.
In our implementation of CCH, we also maintain arrays with each vertex's tail and the unpacking information.
Note that the unpacking information will be temporarily invalid because the referenced edge IDs are for \augEdges and not the reduced graph \minAugEdges we are currently building.
Fixing this is the third pass's goal but requires an explicit edge ID mapping.
For this, we also establish a mapping from old edge IDs to new edge IDs.
This mapping can be computed while the edges are copied to the reduced graph \minAugGraph.
Finally, in the last pass, we apply this mapping to the unpacking information.
In this pass all edges can be processed independently.
Therefore, any kind of parallelization on the edge level will work.\looseness-1

\subsubsection{Parallelization}
\label{sec:cchpp:parallelization}

Parallelization works similar for the basic and perfect customization, which is why we discuss them jointly.
In Sections~\ref{sec:cchpp:basic} and \ref{sec:cchpp:perfect}, we already hinted at the fact that choosing the triangle order to have all edits to the same edge more or less consecutive is helpful for parallelization.
Additionally, we have to make sure that all information necessary to relax a triangle have been computed before the triangle is relaxed.

The original CCH publication~\cite{DBLP:journals/jea/DibbeltSW16} already proposed a solution to this problem.
Unfortunately, their approach required many synchronization steps.
Here, we focus on a task-based parallelization approach introduced in~\cite{DBLP:journals/jea/Buchhold0W19}.
We first explain this approach for the perfect customization and afterwards transfer it to our batched triangle relaxation for the basic customization.

The division of work is based on the separator decomposition and each task is responsible for one subgraph and its top-level separator.
Initially there is only one task and it is responsible for the whole graph.
The key observation here is that the top-level separator splits the graph into two independent subproblems and by construction all vertices in those subgraphs have a lower rank, i.e., should be processed afterwards.
Therefore, a task simply processes all vertices of the top-level separator sequentially and afterwards spawns a new task for each subgraph.
We note that if the size of subgraph is below a certain threshold, it is more efficient to handle the whole subgraph sequentially.
In our reference implementation $\numVertices/(\alphaParameter \cdot \numCores)$ is used as the threshold value, with \numCores being the number of cores and the tuning parameter \alphaParameter set to $32$.
This value was suggested by~\cite{DBLP:journals/jea/Buchhold0W19}.
The main benefit of this approach is that no synchronization is necessary as a task is only spawned when it is ready to be processed.

We now explain how to transfer this task-based approach to the basic customization.
The relevant difference here is that triangles have to be processed bottom up.
Therefore, the tasks need to be processed in reversed order.
To be more precise, we can only process a task after all its children have been completed.
For this, some synchronization is necessary, but Buchhold et al.~\cite{DBLP:journals/jea/Buchhold0W19} observed that the overhead for this approach is much smaller than for the approach in the original CCH publication~\cite{DBLP:journals/jea/DibbeltSW16}.

The batched triangle relaxation algorithm shown in Algorithm~\ref{algo:basic_customization} can also be parallelized with this separator-based approach.
Besides the handling of the tasks there is only one change necessary.
The outer loop in Line~\ref{algo:batched_triangle_basic:main_loop} has to be restrict to the separator of the current subgraph.
Because only weights associated with edges of the current vertex \uVertex are modified, no concurrent modifications can occur.
Note that Algorithm~\ref{algo:basic_customization_buchhold} cannot be adopted that easily since the writes are not local.

\subsection{Queries}

The elimination tree query is the preferred query algorithm for CCHs since it is fastest and no priority queues are required; see Algorithm~\ref{algo:query}.
For each vertex \vVertex, the query maintains two tentative distances $\distArr^{\uparrow}[\vVertex]$ and $\distArr^{\downarrow}[\vVertex]$ representing the shortest distance found so far from $s$ and to $t$, respectively.
Starting at $s$ and $t$, it walks up the elimination tree and processes each vertex on the way.
Processing a vertex means to relax all outgoing edges using the distance functions $\lchu$ and $\lchd$ for searches from $s$ and $t$, respectively.
We note that one cannot simply stop the searches from $s$ and $t$ meet, as one might find a shorter path higher up in the tree.
However, Buchhold et al.~\cite{DBLP:journals/jea/Buchhold0W19} proposed the following optimization for pruning the search.
For this, one maintains a tentative total distance $\mu$ to be the minimum of $\distArr^{\uparrow}[\vVertex] + \distArr^{\downarrow}[\vVertex]$ over all vertices \vVertex.
Then, if a vertex \uVertex is processed for which the tentative distance is already bigger than $\mu$, we do not need to relax the outgoing edges as the shortest path will not contain \uVertex; see Line~\ref{algo:query:check_tent} of Algorithm~\ref{algo:query}.
With this, elimination tree queries are consistently faster than classical CH queries (on CCH) across all distances.

Buchhold et al. present another optimization which accelerates queries further.
Because the search graphs are acyclic, the distance of a vertex will never be read again after its outgoing edges were relaxed.
Thus, if we make sure that the searches from $s$ and $t$ are synchronized, we can immediately reset the tentative distances to $\infty$ after a vertex is processed, as shown in Line~\ref{algo:query:reset}.
A separate distance resetting step after the query -- as proposed in~\cite{DBLP:journals/jea/DibbeltSW16} -- becomes unnecessary.
How the two searches are synchronized can be seen in Lines~\ref{algo:query:interleave} and~\ref{algo:query:interleave2}.
Before the two searches have met, Line~\ref{algo:query:interleave} makes sure that the search that is currently lower catches up.
Line~\ref{algo:query:interleave2} shows how the searches continue together after having met.

\begin{algorithm}[tb]
	\caption{Elimination tree query.}
	\label{algo:query}
	\KwData{$\elimtree[\vVertex]$: the parent of \vVertex in the elimination tree (or $\bot$).}
	\KwData{$\lchu[\uv]$: length of $\uv \in \echu$ in the augmented graph.}
	\KwData{$\lchd[\uv]$: length of $\vu \in \echd$ in the augmented graph.}
	\KwData{$\distArr^{\uparrow}[\vVertex]$: tentative distance of \vVertex in the forward search, initially~$\infty$.}
	\KwData{$\distArr^{\downarrow}[\vVertex]$: tentative distance of \vVertex in the backward search, initially~$\infty$.}
	\SetKwFunction{query}{Query}
	\SetKwFunction{processVertex}{ProcessVertex}
	
	\BlankLine
	\Fn{$\processVertex(\uVertex, \tentDist, \distArr, \mathtt{l})$}{
		\If{$\distArr[\uVertex] < \tentDist$}{
			\label{algo:query:check_tent}
			\For{all edges $\uv \in \echu$}{
				$\distArr[\vVertex] \gets \min(\distArr[\vVertex], \distArr[\uVertex] + \mathtt{l}[\uv])$\;
			}
		}
		$\distArr[\uVertex] \gets \infty$\;
		\label{algo:query:reset}
	}
	
	\BlankLine
	\Fn{$\query(\aSource, \aTarget)$}{
		\While{$\aSource \neq \aTarget$}{
			\tcp{we assume that $\mathtt{rank}[\bot]=\infty$}
			\If{$\mathtt{rank}[\aSource] < \mathtt{rank}[\aTarget]$}{
				\label{algo:query:interleave}
				$\processVertex(\aSource, \infty, \distArr^{\uparrow}, \lchu)$\;
				$\aSource \gets \elimtree[\aSource]$\;
			}
			\Else{
				$\processVertex(\aTarget, \infty, \distArr^{\downarrow}, \lchd)$\;
				$\aTarget \gets \elimtree[\aTarget]$\;
			}
		}
		$\tentDist \gets \infty$\;
		$\uVertex \gets \aSource$\;
		\While{$\uVertex \neq \bot$}{
			\label{algo:query:interleave2}
			$\tentDist \gets \min(\tentDist, \distArr^{\uparrow}[\uVertex] + \distArr^{\downarrow}[\uVertex])$\;
			$\processVertex(\uVertex, \tentDist, \distArr^{\uparrow}, \lchu)$\;
			$\processVertex(\uVertex, \tentDist, \distArr^{\downarrow}, \lchd)$\;
      $\uVertex \gets \elimtree[\uVertex]$\;
			%$\distArr^{\uparrow}[\aSource] \gets \infty$\;
			%$\distArr^{\downarrow}[\aSource] \gets \infty$\;
		}
		\Return{\tentDist}\;	
	}
\end{algorithm}

\section{Extended Scenarios}\label{sec:cchpp:ext_queries}

So far we focused on the basic CCH algorithms which enables us to speed up point-to-point shortest path queries while still allowing updates to the metric.
This chapter introduces solutions for extended scenarios like one-to-many queries, point-of-interest queries and turn costs that are used in common routing applications.

\subsection{Incremental One-to-Many}

The previous sections only considered point-to-point queries.
While this is the most common type of query, it is often desirable to do one-to-many (or many-to-one) queries, where we want to find the shortest path from one source $\aSource$ to many targets $\aTarget_1,\dots,\aTarget_k$.
Here we consider the incremental variant where the source is given in the beginning and the targets are given one by one.
Strasser and Zeitz considered this setting for classical contraction hierarchies and developed an algorithm called \emph{Lazy RPHAST} \cite{DBLP:journals/jea/StrasserZ22}.
It can in principle be applied directly to CCHs.
However, this has the disadvantage that it does not use the improved query based on elimination trees.
Here we propose an adaptation of Lazy RPHAST to CCHs utilizing the elimination tree, which makes the incremental one-to-many setting more efficient and simpler.

To explain our approach based on the elimination tree, we first need to make some observations about the tentative distances computed by the search.
When doing the forward search from \aSource, we compute the tentative distance $\distArr^{\uparrow}[\vVertex]$ for each ancestor \vVertex of \aSource in the elimination tree.
It indicates that there is a path of this length from \aSource to \vVertex.
However, other than with a Dijkstra search, we do not have a guarantee that $\distArr^{\uparrow}[\vVertex]$ will actually become the distance from \aSource to \vVertex at some point.
If we want the distance from \aSource to \vVertex, we would additionally need to do a backward search from \vVertex, which might yield a shorter path via a vertex further up in the elimination tree.

The core idea of our approach is to not only compute these tentative distances but also the actual distance from \aSource to vertices for which we get this information more or less for free.
We memoize this information by storing it in an additional array $\distArr[\vVertex]$.
In the following we first describe how and for which vertices \vVertex we compute $\distArr[\vVertex]$.
Then we discuss how this information is used to speed up subsequent searches for different target vertices.
The result is shown in Algorithm~\ref{algo:lazy_rphast_elim_tree}

Assume we already did the forward search from \aSource yielding the tentative distances $\distArr^{\uparrow}[\vVertex]$ for every ancestor \vVertex of \aSource.
Now consider the first target vertex $\aTarget_1$ and let $v_1, \dots, v_\ell$ be the path from $\aTarget_1$ to the root of the elimination tree with $v_1 = \aTarget_1$.
Observe that for the root $v_\ell$, we get $\distArr[v_\ell] = \distArr^{\uparrow}[v_\ell]$.
Moreover, for $v_{\ell - 1}$, the distance is either $\distArr^{\uparrow}[v_{\ell - 1}]$ or $\distArr[v_\ell]$ plus the length $\lchd[v_{\ell - 1}v_\ell]$ of the edge from $v_{\ell}$ to $v_{\ell - 1}$.
In general, we can observe that for any $\vVertex \in \{v_1, \dots, v_\ell\}$ it holds that $\distArr[v] = \min\big(\distArr^{\uparrow}[v], \min_{vw \in \echu}\{\lchd[\vw] + \distArr[w]\}\big)$.
Thus, instead of doing the backward search from $t_1$, we can process the path from the root down to $t_1$ and compute $\distArr[\vVertex]$ for ever vertex $\vVertex$ on the way.
In Algorithm~\ref{algo:lazy_rphast_elim_tree}, you can observe that we first search upwards in the tree to find the path to the root starting in Line~\ref{algo:lazy_rphast_elim_tree:explore}.
Then, starting in Line~\ref{algo:lazy_rphast_elim_tree:process}, the path is processed from top to bottom, storing the distances from \aSource on the way.

For the second aspect of how we can save time in subsequent queries, we can observe that once the upward search from $\aTarget_i$ reaches a vertex \vVertex with already known value $\distArr[\vVertex]$, we do not go up further as $\distArr[\uVertex]$ is already known for all ancestors of \vVertex.
This is why in Line~\ref{algo:lazy_rphast_elim_tree:explore}, we do not stop at the root but once $\distArr[\vVertex]$ is known.

\begin{algorithm}[tb]	
	\caption{Elimination tree based Lazy RPHAST algorithm.}
	\label{algo:lazy_rphast_elim_tree}
	\KwData{$\elimtree[\vVertex]$: the parent of \vVertex in the elimination tree (or $\bot$).}
	\KwData{$\distArr^{\uparrow}[\vVertex]$: distance from \aSource to \vVertex as computed by the the forward search (or $\infty$).}
	\KwData{$\lchd[\uv]$: length of $\uv \in \echd$ in the augmented graph.}
	\KwData{$\distArr[\vVertex]$: memoized distance from \aSource to \vVertex, initially~(or $\bot$).}
	
	\SetKwFunction{LazyRPHAST}{LazyRPHAST}	
	\BlankLine
	\Fn{\LazyRPHAST{\uVertex}}{
		\tcp{Determine all vertices for which $\distArr$ needs to be computed}
		$\vVertex \gets \uVertex$\;
		$\aStack \gets []$\;
		\While{$\distArr[\vVertex] = \bot$}{\label{algo:lazy_rphast_elim_tree:explore}
			$\aStack \gets \aStack + [\vVertex]$\;
			\If{$\elimtree[\vVertex] = \bot$}{
				\Break\;
			}
			$\vVertex \gets \elimtree[\vVertex]$\;
		}
		\tcp{Compute \distArr for all those vertices}
		\For{\vVertex in \aStack in reverse order}{\label{algo:lazy_rphast_elim_tree:process}
			$\distArr[\vVertex] \gets \distArr^{\uparrow}[\vVertex]$\;
			\For{all edges $\vw \in \echu$}{
				$\distArr[\vVertex] \gets \min(\distArr[\vVertex], \lchd[\vw] + \distArr[\wVertex])$\;
			}
		}
		\Return{$\distArr[\uVertex]$}\;
	}
\end{algorithm}

Note that incremental many-to-one (instead of one-to-many) queries can be handled symmetrically.

\subparagraph{CCH-Potentials.}

In extended routing scenarios, it can be relevant to enable different search variants (e.g., with or without turn cost).
If there are too many variants, it might be infeasible to create one CCH for each variant.
In this case, one can potentially still use a single CCH to speed up related searches.
The core idea is to us the distances in the CCH as potential function for an A* search.

For this, note that if A* searches from \aSource to \aTarget, it requires a potential value for every vertex \vVertex it visits.
This potential value should be lower bound on the distance form \vVertex to the target \aTarget.
The idea now is to query the distance from \vVertex to \aTarget in the CCH and use the result as a lower bound.
For this to work, one has to make sure that the distances in the CCH are actually lower bounds for the desired distances (which is, e.g., the case when the A* search computes distances with turn costs while the CCH does not include turn costs).
Observe that in this scenario, the CCH queries always involve the same target \aTarget and one can thus utilize lazy RPHAST for efficient potential computations.
We call this approach \emph{CCH-Potentials}.
We note that this has been considered before for CHs~\cite{DBLP:journals/jea/StrasserZ22} but not yet for CCHs.

Beyond supporting various extended scenarios with the same CCH, we note that the main benefit of this approach is that adjusting Dijkstra's algorithm to new scenarios is generally much easier than adjusting the CCH itself.

\subsection{Point of Interest Queries}

The next scenario we want to consider are \emph{point-of-interest} (POI) queries.
This queries allow, for example, to find gas stations close to the current location.
We formalize this as the $k$-nearest-neighbor problem.
For this, we are given not only the graph but also a set of targets $\targets \subseteq \vertices$.
Then, for a given source \aSource, the goal is to find the $k$ targets in \targets that are closest \aSource.

For this scenario, Buchhold and Wagner \cite{DBLP:conf/wea/BuchholdW21} introduced a light-weight \emph{selection} phase between the customization phase and the query phase, which can be used to do preprocessing based on \targets.
The more interesting part of the algorithm by Buchhold and Wagner is the adjusted query phase for the $k$-nearest-neighbor setting.
In the following, we first give the core idea of the query and afterwards briefly discuss the selection phase.
We will see that the query fits into the one-to-many setting and can be accelerated with Lazy RPHAST.
Our experiments in Section~\ref{sec:cchpp:eval} show that this results in speedups of up to an order of magnitude.

The core idea of the query by Buchhold and Wagner~\cite{DBLP:conf/wea/BuchholdW21} is based on the separator hierarchy; recall Lemma~\ref{lem:separator-hierarchy-from-elim-tree}.
The separators hierarchy can be viewed as a tree in which each node represents a subgraph called cell.
The cell corresponding to the root is the whole graph and the cell of each non-root node is a subgraph of its parent's cell.
The query now traverses this tree from the root, maintaining a set of the $k$ closest targets found so far.
Before entering a cell \aCell, the query computes a lower bound for the distance between \aSource and vertices in \aCell.
If the resulting distance is larger than the $k$th closest target currently known, we can prune the search as we cannot find a better target in this cell.
Thus, to be able to prune many cells, it makes sense to heuristically first visit cells that contain target vertices close to \aSource.

This query requires multiple shortest path computations all starting at \aSource, which we can improve using Lazy RPHAST.
Moreover, one regularly has to compute the distance of \aSource to all targets in a separator.
To enable this, we need to be able to quickly compute the intersection of a separator with the targets \targets.
This is where the selection phase comes into play.
Buchhold and Wagner~\cite{DBLP:conf/wea/BuchholdW21} propose to just create an array of size $n$ storing for each vertex the information whether it is a target or not.
Then one can get all targets in a separator by iterating the separator once.
As the set of targets \targets is typically small, we propose an alternative approach.
We propose to sort \targets.
Then, for a given separator $S$, we observe that $S$ is a path of vertices with just one child in the elimination tree; see Lemma~\ref{lem:separator-hierarchy-from-elim-tree}.
Thus, the indices of vertices in $S$ are consecutive and we can find $S \cap \targets$ by doing a binary search on the sorted targets for the lowest ID in $S$ and then iterating linearly through \targets until the first vertex in \targets that has a higher ID than the highest ID in $S$.

\subsection{Turn Costs and Restrictions}

So far we always assumed that a graph is given and did not consider what vertices and edges in this graph actually represent.
For road networks a natural assumption is that edges represent road segments and vertices represent crossing.
However, this model ignores turn costs and turn restrictions, which are critical for realistic routing.
In the literature two models have been considered to represent turn costs, but it turns out that for CCH the only competitive solution is to model turn costs into the graph~\cite{DBLP:conf/atmos/BuchholdWZZ20}.
We call graphs with modeled turn costs expanded graphs.

On expanded graphs the CCH algorithms can be applied without any modification.
Unfortunately, this results in a surprisingly large slowdown for both preprocessing steps.
The slowdown is much larger than the factor by which the graph size increased.
Buchhold et al.~\cite{DBLP:conf/atmos/BuchholdWZZ20} investigated the reason for this and proposed various changes to mitigate this.
On one hand they created a variant of Inertial Flow Cutter tailored to compute rank orders for expanded graphs more efficiently.
On the other hand they introduced additional steps to the metric independent preprocessing that reduce the number of edges in the augmented graph.
In combination their result showed that the slowdown of the preprocessing can be reduced to an factor roughly matching the expansion factor without loosing query efficiency.
For more details, we refer to~\cite{DBLP:conf/atmos/BuchholdWZZ20}.

\section{Evaluation}\label{sec:cchpp:eval}
In this section we present the experimental evaluation of our reference implementation.
After describing our setup, we start in Section~\ref{eval:performance} with a comparison of our implementation with two previous implementations~\cite{DBLP:journals/jea/DibbeltSW16,DBLP:journals/jea/Buchhold0W19}.
For this we consider all three phases individually and especially evaluate the impact of our proposed changes.
In Section~\ref{eval:extended} we take a closer look at the extended scenarios.
First of, we evaluate the laze RPHAST algorithm on its own and afterwards its impact on the point of interest queries.
In the last section~\ref{eval:related} we compare our reference implementation of CCHs with other routing algorithms.
For this we focus on the main competition CRP~\cite{DBLP:journals/transci/DellingGPW17} but also consider the non-customizable predecessor, i.e., CH~\cite{DBLP:journals/transci/GeisbergerSSV12}.

\subparagraph{Environment.}
Our primary benchmark machine runs openSUSE Leap 15.3 (kernel 5.3.18), and has 64\,GiB of DDR3-1600 RAM and two Intel Xeon E5-2670 CPUs, each of which has eight cores clocked at 2.6\,Ghz and 8~$\times$~64\,KiB of L1, 8~$\times$~256\,KiB of L2, and 20\,MiB of shared L3 cache.
We use this machine to ensure the comparability of our results to previous works~\cite{DBLP:journals/jea/DibbeltSW16,DBLP:journals/jea/Buchhold0W19} which were also evaluated on this computer.
Additionally, this machine allows a rough comparison with the results on CRP presented in~\cite{DBLP:journals/transci/DellingGPW17}\footnote{
	According to the comparison methodology from~\cite{DBLP:series/lncs/BastDGMPSWW16} (see \url{https://i11www.iti.kit.edu/~pajor/survey/}), the machine used in~\cite{DBLP:journals/transci/DellingGPW17} (SPA-2) is about 20\% slower than ours (compute11).
	We reran the benchmark for our machine and obtained a score of 38\,914\,ms, which is somewhat slower than the previously reported 36\,582\,ms.
	This is likely due to the mitigations for side-channel attacks utilizing speculative execution such as Meltdown and Spectre.
	SPA-2 would only be about 12\% slower than our machine compared to this updated score.
	Generally, these scaling factors have to be interpreted very carefully.
	They are obtained from one-to-all Dijkstra searches on continental-sized road networks.
	This benchmark heavily emphasizes memory bandwidth while neglecting other critical factors such as CPU frequency, cache size and cache speed.
	Also, note that the authors of~\cite{DBLP:journals/transci/DellingGPW17} even used our machine to evaluate their turn-aware CH implementation, stating that it achieves ``comparable'' performance.
	We conclude that our machine and SPA-2 yield running times in the same order of magnitude and that our machine is probably, in most cases, slightly faster.
}.
Since the machine is already almost a decade old at the time of writing, we expect our results to be a reasonably conservative approximation of the performance that can be expected when deploying these algorithms today.

We implemented our algorithms in Rust\footnote{
	The code for this paper and all experiments is available at \url{https://github.com/kit-algo/cchpp}
} and compiled them with \texttt{rustc 1.64.0-nightly (830880640 2022-06-28)} in the release profile with the \texttt{target-cpu=native} option.
For the computation of rank orders, we use InertialFlowCutter (IFC)\footnote{\url{https://github.com/kit-algo/InertialFlowCutter}}~\cite{DBLP:journals/algorithms/GottesburenHUW19}.
As shown in the extensive evaluation in~\cite{DBLP:journals/algorithms/GottesburenHUW19}, IFC currently achieves the best rank order quality while taking only about two times as long as InertialFlow~\cite{DBLP:conf/wea/SchildS15}, the fastest approach with good quality.
We extend IFC with the computation of cut-based orders and the reordering of separator vertices.
These extensions were published as a pull request on GitHub\footnote{\url{https://github.com/kit-algo/InertialFlowCutter/pull/6}} and have since been merged into the project.
InertialFlowCutter is written in C++ and were compiled with GCC~10.3.0 using optimization level 3.

\subparagraph{Inputs.}
We use our main benchmark instances for time-independent route planning: DIMACS Europe and OSM Germany.
We use synthetic turn costs of 100\,s for U-Turns and free turns otherwise for DIMACS Europe as suggested in~\cite{DBLP:journals/transci/DellingGPW17}.
For OSM Germany, we extract real-world turn restrictions from the OSM data.
Additionally, we include a city-scale network of the Stuttgart region provided by PTV.
This network includes real-world production-grade turn cost and restriction data.
It has roughly 110\,k vertices and 252\,k edges, which makes it about two orders of magnitude smaller than the Germany and Europe instances.
Our evaluation in~\cite{DBLP:conf/atmos/BuchholdWZZ20} included additional city-scale networks.

\subparagraph{Methodology.}
For the computation of the rank order, we perform the partitioning ten times and report the average running time.
Contraction and customization running times are obtained as averages over 100 runs.
For the queries, we perform 1\,000\,000 point-to-point queries where both source and target are vertices drawn uniformly at random.
We utilize parallelization for all phases except the queries.

\subsection{CCH Performance}\label{eval:performance}
In this section we want to compare our improved CCH implementation with previous implementation from Dibbelt et al.~\cite{DBLP:journals/jea/DibbeltSW16} and Buchhold et al.~\cite{DBLP:journals/jea/Buchhold0W19}.
For this we analyze three phases individually in the next three subsections.
We will also analyze some parts of each phase in greater detail to show the effects of our contributions.

\subsubsection{Metric-Independent Preprocessing}
In this subsection we analyze the metric-independent preprocessing with a special focus on the contraction algorithm.
Table~\ref{tab:cchpp:preprocessing} shows running times of the metric-independent preprocessing.
The ordering, performed by IFC, is two orders of magnitude slower than the contraction (with chordal completion) and thus dominates the running time.
Still, computing the rank order of a continental-sized network takes only about six minutes, thanks to the efficient parallelization of IFC.
This is faster than the preprocessing time of a classical non-customizable CH~\cite{DBLP:journals/transci/GeisbergerSSV12}\footnote{
	The CH preprocessing is typically performed sequentially, which makes this a somewhat unfair comparison.
	Parallelization approaches for CH preprocessing have been described in the context of time-dependent CH~\cite{vetter2009parallel,DBLP:journals/jea/BatzGSV13}.
	However, these approaches do not scale very well in the classical setting.
	Further, to the best of our knowledge, there is neither a publication on parallelized CH preprocessing in the time-independent case nor any open source implementation.
}.
We only have running times of competing contraction algorithms for Europe\footnote{
	Note that the orders used in~\cite{DBLP:journals/jea/DibbeltSW16} were obtained with KaHiP, which finds slightly worse orders than IFC.
	However, according to~\cite{DBLP:journals/algorithms/GottesburenHUW19}, the advantage for our algorithms from this should be at most 10\%.
}.
Still, the speedups are so significant that we can safely conclude that the chordal completion algorithm is the best approach for the contraction.
Chordal completion is two orders of magnitude faster than the naive baseline, one order of magnitude faster than the engineered contraction graph approach, and also much simpler to implement\footnote{See \url{https://github.com/RoutingKit/RoutingKit/pull/75/commits/16de474b2c3} where we replace the contraction graph approach in RoutingKit with the chordal completion algorithm.}.
On Stuttgart, the running time is in the single-digit milliseconds.
In practice, the running time of the chordal completion is so fast that it disappears behind memory allocation for the graph data, reconstruction of the separator decomposition and other setup/management operations.

\begin{table}
	\centering
	\caption{
		Running time in seconds of the metric-independent preprocessing algorithms.
		Our rank orders were computed with IFC.
		For contraction, we compare our own chordal completion algorithm against the contraction graph approach and the naive baseline as reported in~\cite{DBLP:journals/jea/DibbeltSW16}.
		Our results and the numbers from~\cite{DBLP:journals/jea/DibbeltSW16} were obtained on the same machine.
		The total running time includes the ordering, the contraction with our chordal completion approach and additionally the reconstruction of the elimination tree and separator decomposition, modifying the order into an DFS post order and all other setup operations.
	}\label{tab:cchpp:preprocessing}
	\input{table/preprocessing.tex}

\end{table}

\subsubsection{Customization}

We now evaluate the performance of the customization, with a focus on our improvements to the parallelization; see Table~\ref{tab:cchpp:customization}.
Considering the total running time, we observe that our batched triangle relaxation-based customization in the sequential case is roughly two times faster than the results observed in~\cite{DBLP:journals/jea/Buchhold0W19} and about four times faster than the baseline~\cite{DBLP:journals/jea/DibbeltSW16}.
With full parallelization with 16 threads, the picture remains similar on Europe.
However, on the smaller Stuttgart instance, the difference to the approach from~\cite{DBLP:journals/jea/Buchhold0W19} becomes much smaller.
Interestingly, the performance differences are fueled primarily by the construction step, where our approach is sequentially four to five times faster than~\cite{DBLP:journals/jea/Buchhold0W19}.
One important reason is that the batched triangle relaxation allows us to record unpacking information during the basic customization without any synchronization issues.
Thus, in contrast to~\cite{DBLP:journals/jea/Buchhold0W19}, we do not need to enumerate lower triangles while constructing the minimal augmented graph.
We observe that our algorithms utilize additional threads reasonably well, as long as the instances are sufficiently large.
On Europe, the total speedup with eight threads is 5.8 and 8.6 with 16 threads.
In contrast, on Stuttgart, our best achieved speedup is 4.2 with eight threads.
Adding more threads even starts to degrade the performance gradually.
In terms of absolute numbers, our approach enables customization of continental-sized instances in about 10\,s sequentially and a little more than a second fully parallelized.
These improvements finally bring CCH customization running times into a similar range to CRP customization running time (10.55\,s sequentially, 1.05\,s with 12 cores~\cite{DBLP:journals/transci/DellingGPW17}).
CRP customization times can even be outperformed when only applying the basic customization.

\begin{table}
	\centering
	\setlength{\tabcolsep}{4pt}
	\caption{
		Running times by number of threads of different steps of the customization on Stuttgart and Europe in comparison with the baseline results reported in~\cite{DBLP:journals/jea/DibbeltSW16} and the improvements proposed in~\cite{DBLP:journals/jea/Buchhold0W19}.
		Our results and the numbers from~\cite{DBLP:journals/jea/DibbeltSW16,DBLP:journals/jea/Buchhold0W19} were obtained on the same machine.
		Note that the orders used in~\cite{DBLP:journals/jea/DibbeltSW16,DBLP:journals/jea/Buchhold0W19} were obtained with KaHiP and InertialFlow, respectively which find slightly worse orders than IFC.
		However, according to~\cite{DBLP:journals/algorithms/GottesburenHUW19}, the advantage for our algorithms from this should be at most 10\%.
	}\label{tab:cchpp:customization}
	\input{table/customization.tex}

\end{table}

\subsubsection{Queries}

Here, we analyze the query search space statistics for different networks and weight functions.
Buchhold et al. already showed that elimination tree queries with all optimizations proposed by them outperform Dijkstra-based CCH queries across all query distances~\cite{DBLP:journals/jea/Buchhold0W19}.
Thus, we focus on elimination tree queries and do not evaluate Dijkstra-based queries.
Table~\ref{tab:cchpp:queries} depicts our results, and shows search space statistics and running times for various scenarios.

We observe that with only the basic customization, the number of relaxed edges and the distance query running times are very robust against different weight functions.
The only changes are due to the pruning criterion~\cite{DBLP:journals/jea/Buchhold0W19} which sometimes allows skipping the relaxation of some edges.
However, as we use random, i.e., primarily long-range queries for this experiment, this happens only seldomly and has little influence on the results.
Despite Germany being a slightly smaller graph, queries take somewhat longer (around 440\,$\mu$s compared to around 300\,$\mu$s on Europe).
This is because the greater modelling detail of the OSM-based Germany instance corresponds to the number of vertices on the obtained paths and the search space size.
Path unpacking times also correlate with the number of vertices on the result paths, are in the same order of magnitude as the distance computation times and are often slightly faster.

Running the perfect customization improves query performance significantly in all aspects.
The number of relaxed edges is at least halved and the query running times are also roughly halved.
Path unpacking times also improve, but not as much.
This is because, for path unpacking, the work to perform remains the same, and only the cache locality improves.
These improvements are more significant for weight functions with a more pronounced hierarchy.
The different weight functions on Germany illustrate this clearly.
Stuttgart behaves differently because it is a city network that is much less hierarchical.

\begin{table}
	\centering
	\setlength{\tabcolsep}{3pt}
	\caption{
		Search space statistics and running times for elimination tree queries on different graphs and weight functions.
		We evaluate queries on \augGraph with the basic customization and on \minAugGraph after performing the perfect customization.
		The number of visited vertices remains the same because elimination tree queries always traverse the full path to the root.
		The number of edges indicates the combined number of edges relaxed in both directions.
		The final column contains the number of vertices on the unpacked shortest path.
		All numbers are averages over 1\,000\,000 random queries.
	}\label{tab:cchpp:queries}
	\input{table/queries.tex}

\end{table}

\subparagraph{Unpacking.}
A look at table~\ref{tab:cchpp:queries} shows that we report significantly faster unpacking times compared to the numbers reported in~\cite{DBLP:journals/jea/DibbeltSW16}.
There, an unpacking time of 253\,$\mu$s is reported for Europe with travel times and 524\,$\mu$s with geo distances (without perfect customization).
Our unpacking times are faster because we store explicit edge unpacking information with each edge while the implementation of~\cite{DBLP:journals/jea/DibbeltSW16} enumerates lower triangles for unpacking.
However, maintaining this unpacking information during the customization costs some performance, specifically in the \minAugGraph construction step.
Thus, there is a trade-off between path unpacking times and customization times.
If path unpacking times are not critical and a slowdown of a factor of two is tolerable, the customization times could be accelerated further.

\subsection{Extended Scenarios}
\label{eval:extended}

In the following, we investigate how CCH performs in the extended problem settings.
For this we first investigate the lazy RPHAST algorithm -- our incremental one-to-many algorithm -- on its own.
Afterwards, we take a look at the point-of-interest queries which are based on lazy RPHAST.

\subsubsection{Lazy RPHAST}

To evaluate our elimination tree-based Lazy RPHAST implementation we do the following.
First, we pick 100 centre vertices uniformly at random.
From each of these, we obtain a ball of vertices \aBall by running Dijkstra's algorithm until the desired number of vertices have been settled.
We then pick 100 source and $2^{14}$ target vertices uniformly at random from \aBall and incrementally compute the distances to all targets from each source.
Figure~\ref{fig:et_vs_dfs} depicts the average running times of the incrementally queried distances on Europe.

Lazy RPHAST on CCH is somewhat slower than on CH.
This is because CCH has a denser search space than CH.
On large ball sizes, the slowdown goes up to a factor of 1.5.
However, on smaller ball sizes, the slowdown becomes less significant.
The elimination tree-based variant even becomes marginally faster than the CH variant on the two smallest ball sizes.
On CCH, we observe a minor advantage for the elimination tree-based realization due to the more efficient implementation.
However, there is one exception where the elimination tree variant is less efficient:
When \aBall is sufficiently large, and many targets are queried, the DFS-based variant becomes marginally faster.
This is because the path to the elimination tree root might contain vertices which are not reachable in the CCH search space because some edges were removed during the perfect customization.
We confirmed this experimentally by running the algorithms without perfect customization on \augGraph.
In this case, the DFS traverses the same vertices as the elimination tree-based variant, and the effect disappears.
Usually, however, the number of vertices unnecessarily visited by the elimination tree variant is small compared to the total work.
It is thus more than compensated by the more efficient implementation.
We, therefore, conclude that the elimination tree-based Lazy RPHAST variant is the better default choice on CCH.
Regarding absolute numbers, the average time per distance query to compute starts in a similar range (around 0.1\,ms) and quickly becomes faster as more distances are queried.
When computing more than $2^{10}$ distance from the same source on a small ball, the average time per distance even becomes faster than Hub Label queries (roughly 0.5\,$\mu$s).
We conclude that applying the many-to-one variant of Lazy RPHAST on CCH allows us to carry over all the results from Strasser et al. \cite{DBLP:conf/wea/StrasserZ21} to a customizable setting at the cost of just a minor slowdown.

\begin{figure}
	\centering
	\includegraphics[width=\linewidth]{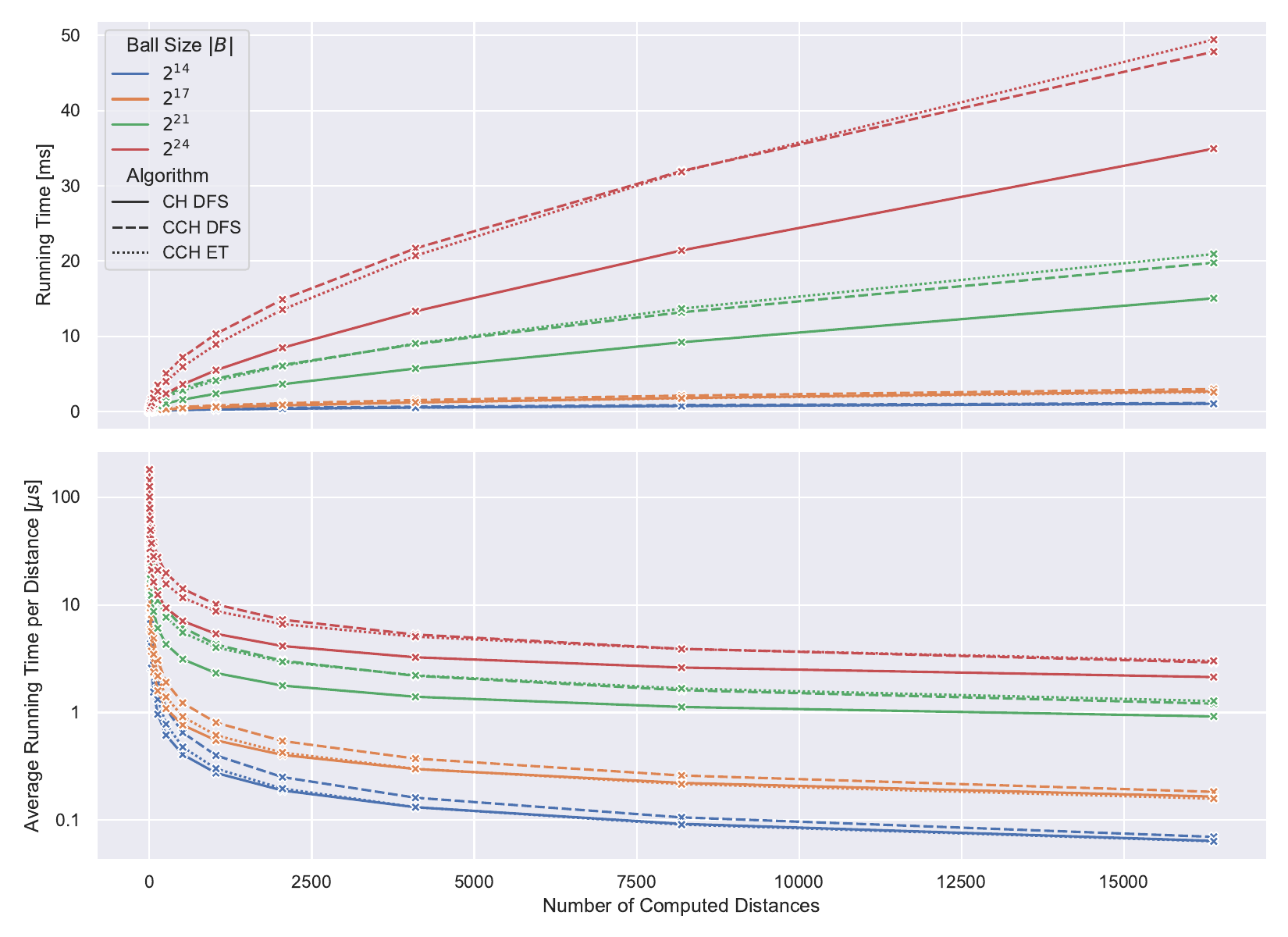}
	\caption{
		Average running times of Lazy RPHAST on CH and CCH while incrementally querying $|T| = 2^{14}$ targets from a ball of varying size $|\aBall|$ on Europe, excluding selection times.
		The upper figure contains the total elapsed running time.
		The lower figure contains the averaged running time per source, i.e., $y/x$ from the upper figure.
		Note the different y-axis scales and units.
	}\label{fig:et_vs_dfs}
\end{figure}

\subsubsection{Point-of-Interest Queries}

We now evaluate CCH-based algorithms to answer POI queries efficiently.
Figure~\ref{fig:knn} depicts mean running times for finding the four closest targets from a set of varying size with the different algorithms.
We drew $100$ target sets uniformly at random, performed the selection, and executed nearest neighbor queries from $100$ source vertices drawn uniformly at random for each target set.
There, we also present a direct comparison between the original separator-based algorithm and our optimized version.
Here, we focus on evaluating which CCH-based algorithm is most suitable in different scenarios.

We observe that the performance of the baseline, Dijkstra's algorithm, strongly depends on the number of POIs.
Since targets are drawn uniformly at random, having more targets corresponds to the closest targets being much closer.
With $2^{11}$ or more targets, Dijkstra's algorithm achieves interactive query times and may present a feasible, practical option.
The selection phase for Dijkstra's algorithm only consists of marking the targets so it can be quickly determined at query time if a settled vertex is a target.
Thus, it is extremely fast.
Lazy RPHAST, in contrast, does not have a selection phase and computes distances to all targets without any stopping criterion.
Thus, it is fast for small target sets.
The bucket query approach~\cite{DBLP:conf/alenex/KnoppSSSW07} (here called BCCH) has the fastest queries (well below 100\,$\mu$s) across all target set sizes.
Further, query times profit from larger target sets.
However, the selection phase is relatively slow and goes into the minutes for large target sets.
Thus, the usefulness of BCCH depends on the scenario.
It is the best choice in an offline scenario where targets are known in advance.
In an online scenario where the targets are part of the query (for example, a user types ``burger restaurants'' in the search field of his maps application), BCCH is not feasible.
Finally, the separator-based algorithm presents a very robust trade-off.
Selection times (consisting of sorting the targets by CCH rank) are barely slower than the selection for Dijkstra's algorithm.
Query times are consistently below 1\,ms, which is more than sufficient for typical map applications.
Therefore, the separator-based algorithm is the strongest contender in online scenarios.
Even though these query times are an order of magnitude slower than BCCH, the absolute times are so fast that the algorithm is also a reasonable choice in offline scenarios.
Note that our optimizations are crucial in making the separator-based algorithm competitive for the online scenario.
Without them, the selection takes 10-20\,ms~\cite{DBLP:conf/wea/BuchholdW21}.
Thus, the separator-based algorithm would be dominated by Lazy RPHAST for few targets and Dijkstra's algorithm for many targets.

\begin{figure}
	\centering
	\includegraphics[width=\linewidth]{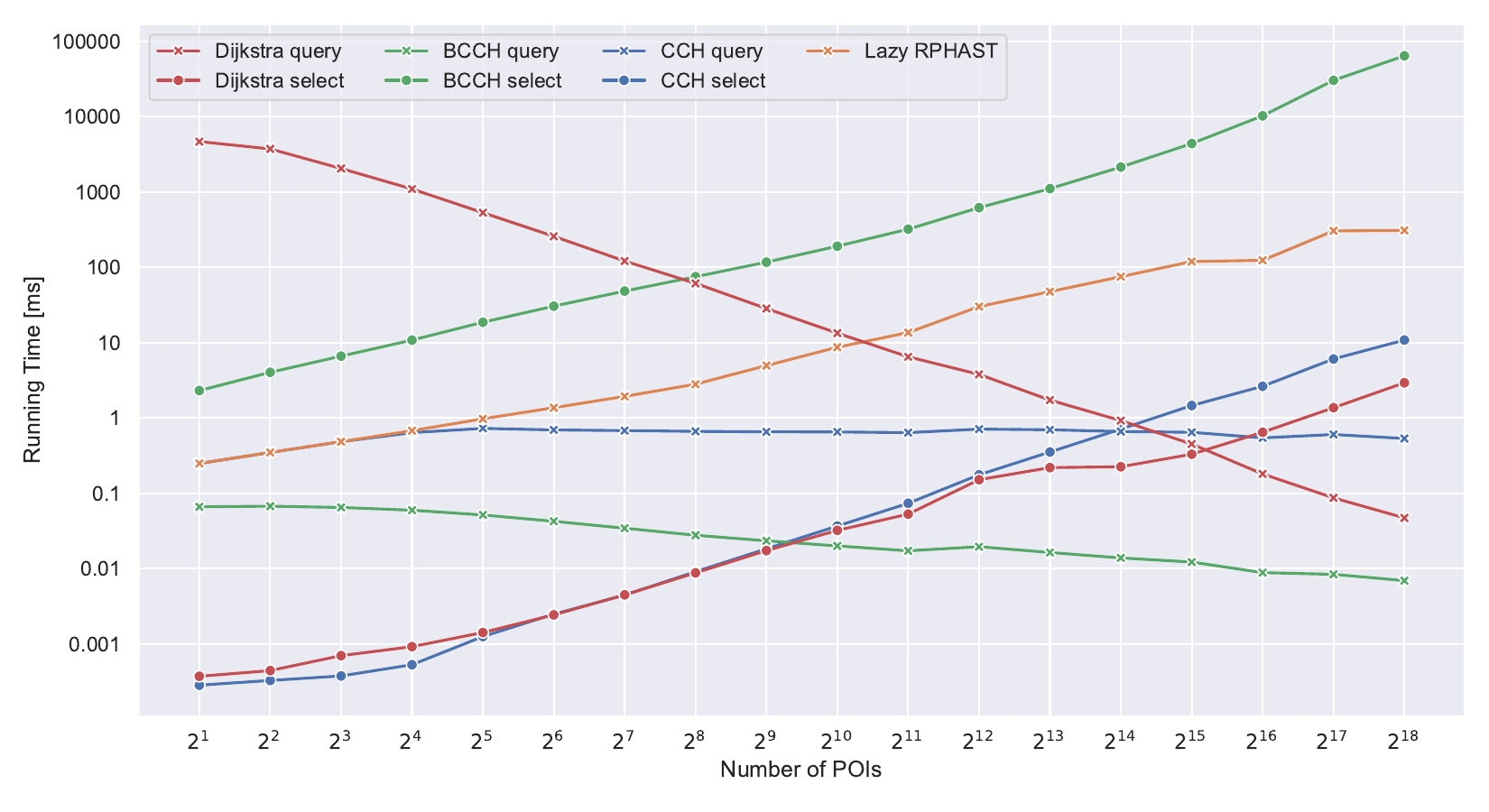}
	\caption{
		Average running times for different nearest-neighbor algorithms on Europe to find the $k=4$ closest targets from a POI set of varying size.
	}\label{fig:knn}
\end{figure}

\subsection{Related Approaches}\label{eval:related}
In this section we compare the performance of our reference CCH implementation with related routing algorithms.
Table~\ref{tab:cchpp:related_work} shows our results on the most prominent research benchmark instance, the DIMACS Europe graph.
Dijkstra's algorithm, the non-accelerated baseline, requires no preprocessing but has prohibitively slow query running times.
CH~\cite{DBLP:journals/transci/GeisbergerSSV12}, the non-customizable predecessor and foundation to CCH, achieves fast queries on travel times.
In this case, the (parallelized) preprocessing is 3.4 times faster than the CCH preprocessing, and queries are also marginally faster.
However, CH is not robust against weight functions with a weaker hierarchy.
Both preprocessing and query times degrade significantly when applied to a geo distance weight function.
In contrast, CRP~\cite{DBLP:journals/transci/DellingGPW17} is very robust against different weight functions and can introduce arbitrary weight functions with a customization taking only around a second.
However, queries are somewhat slower than CH queries.
With CCH, we achieve robustness against different weight functions while retaining most of the CH query performance.
The optimizations proposed in this work and~\cite{DBLP:journals/jea/Buchhold0W19} accelerate the CCH customization such that the basic customization is slightly faster than the CRP customization.
Queries with only the basic customization are around five times faster than CRP queries.
For even faster queries, the perfect customization can be applied.
Then, the total customization time is marginally slower than CRP, but queries are an order of magnitude faster and roughly as fast as classical CH queries.
Nevertheless, CRP still has some advantages over CCH.
First, the space required per customized weight function is smaller for CRP.
Second, CRP has been augmented to more extended problem settings and so far appeared more flexible.

\begin{table}
	\centering
	\caption{
		Running times for Dijkstra's algorithm, CH, CRP and CCH on Europe with different weight functions.
		Preprocessing and customization were executed in parallel, queries sequentially.
		For CH and CRP we list unscaled results as reported in \cite{DBLP:journals/transci/DellingGPW17}.
	}\label{tab:cchpp:related_work}
	\begin{tabular}{lrrrrrrrr}
		\toprule
                                                  & \multicolumn{3}{c}{Travel time} & \multicolumn{3}{c}{Geo distance}                                 \\
		\cmidrule(lr){2-4}                                             \cmidrule(lr){5-7}
                                                  & Prepro.                         & Custom. & Queries   & Prepro.              & Custom. & Queries   \\
                                                  & [s]                             & [s]     & [ms]      & [s]                  & [s]     & [ms]      \\
		\midrule
		Dijkstra                                      & --                              & --      & 2\,359.14 & --                   & --      & 1\,972.93 \\
		CH \cite{DBLP:journals/transci/DellingGPW17}  & 109                             & --      & 0.11      & 726                  & --      & 0.87      \\
		CRP \cite{DBLP:journals/transci/DellingGPW17} & 654                             & 1.05    & 1.65      & 654                  & 1.04    & 1.91      \\
		\addlinespace
		CCH \augGraph                                 & \multirow{2}{*}{367}            & 0.58    & 0.30      & \multirow{2}{*}{367} & 0.58    & 0.30      \\
		CCH \minAugGraph                              &                                 & 1.25    & 0.14      &                      & 1.25    & 0.18      \\
		\bottomrule
	\end{tabular}
\end{table}

\section{Conclusion}
\label{sec:cchpp:conclusion}

With this work, we given an accessible introduction into Customizable Contraction Hierarchies and review the state of the art.
Our algorithmic contributions include a novel contraction algorithm, simplification and improvements of customization algorithms, the batched triangle relaxation as an alternative customization method, an elimination tree-based Lazy RPHAST variant, some improvements for the separator-based CCH nearest neighbor search.
Further, we provide an extensive experimental evaluation which demonstrates that CCH can be used to build a comprehensive and competitive routing framework.

\bibliography{References}
\end{document}

%% file: Graphics/Completion.tex
\newcommand{\completionScale}{0.7cm}
\begin{tikzpicture}
  [x=\completionScale,
   y=\completionScale,
   junction/.style={fill=KITblack!5,draw, circle, inner sep=0pt, minimum size=0.6cm},
   edge/.style={>=stealth, shorten >=1pt, shorten <=1pt},
   original edge/.style={edge, draw=black},
   faded/.style={KITblack!50},]

  \node [junction] (v1) at (1,0) {$\vVertex$};
  \node [junction] (w11) at (-2,1) {$\wVertex_1$};
  \node [junction] (w12) at (-1,2) {$\wVertex_2$};
  \node [junction] (w13) at (0,3) {$\wVertex_3$};
  \node [junction] (w14) at (1,4) {$\wVertex_4$};    
  \node [faded,junction] (u12) at (-0.5,-1) {$\uVertex_2$};
  \node [faded,junction] (u11) at (2,-1.5) {$\uVertex_1$};  
  
  \draw (v1) edge [original edge] (w11);
  \draw (v1) edge [original edge] (w12);
  \draw (v1) edge [original edge] (w13);
  \draw (v1) edge [original edge] (w14);

  %\draw (w11) edge [original edge,bend left=40] (w22);
  %\draw (w11) edge [original edge,bend left=50,looseness=1.1] (w23);
  \draw (w11) edge [original edge,bend left=60,looseness=1.2] (w14);
  
  \draw (w12) edge [original edge,bend left=40] (w13);
  \draw (w12) edge [original edge,bend left=50,looseness=1.1] (w14);
  
  %\draw (w13) edge [original edge,bend left=40] (w14);
  
  \draw (u12) edge [original edge,faded] (v1);
  \draw (u12) edge [original edge,faded] (w11);
  \draw (u11) edge [original edge,faded] (v1);
  
  \begin{pgfonlayer}{background}
  	%\draw[dashed,rounded corners=\completionScale] (0.5,1.5) {[rounded corners=3*\completionScale] -- (3,5) -- (-3,5) -- (-3,-1)} -- cycle;
  	\draw[dashed] plot[smooth cycle, tension=1.3] coordinates{(-2,0) (2,4) (-2.14,4.14)};
  \end{pgfonlayer}
  
  \node [junction] (v2) at (1+9,0) {$\vVertex$};
  \node [junction] (w21) at (-2+9,1) {$\wVertex_1$};
  \node [junction] (w22) at (-1+9,2) {$\wVertex_2$};
  \node [junction] (w23) at (0+9,3) {$\wVertex_3$};
  \node [junction] (w24) at (1+9,4) {$\wVertex_4$}; 
  \node [faded,junction] (u22) at (-0.5+9,-1) {$\uVertex_2$};
  \node [faded,junction] (u21) at (2+9,-1.5) {$\uVertex_1$};   
  
  \draw (v2) edge [original edge] (w21);
  \draw (v2) edge [original edge] (w22);
  \draw (v2) edge [original edge] (w23);
  \draw (v2) edge [original edge] (w24);
  
  \draw (w21) edge [original edge,bend left=40] (w22);
  \draw (w21) edge [original edge,bend left=50,looseness=1.1] (w23);
  \draw (w21) edge [original edge,bend left=60,looseness=1.2] (w24);
  
  \draw (w22) edge [original edge,bend left=40] (w23);
  \draw[line width=0.09cm] (w22) edge [original edge,white,bend left=50,looseness=1.1] (w24);
  \draw (w22) edge [original edge,bend left=50,looseness=1.1] (w24);
  
  \draw (w23) edge [original edge,bend left=40] (w24);
  
  \draw (u22) edge [original edge,faded] (v2);
  \draw (u22) edge [original edge,faded] (w21);
  \draw (u21) edge [original edge,faded] (v2);
  
  \begin{pgfonlayer}{background}
  	\draw[dashed] plot[smooth cycle, tension=1.3] coordinates{(-2+9,0) (2+9,4) (-2.14+9,4.14)};
  \end{pgfonlayer}
  
  \draw[->, original edge] (-4.5,-1.8) -- node[anchor=south,sloped] {rank} (-4.5,5);
\end{tikzpicture}

%% file: Graphics/Triangle.tex
\begin{tikzpicture}
  [junction/.style={fill=KITblack!5,draw, circle, inner sep=0pt, minimum size=0.66cm},
   point/.style={draw, circle, inner sep=0pt, minimum size=1ex},
   edge/.style={>=stealth, shorten >=1pt, shorten <=1pt},
   original edge/.style={edge, draw=black},
   edge label/.style={inner sep=1pt,font=\tiny}]

  \node [junction] (u1) at (0,0) {$\uVertex$};
  \node [junction] (v1) at (-0.75,1.5) {$\vVertex$};
  \node [junction] (w1) at (1,2.5) {$\wVertex$};
  
  \draw (u1) edge [original edge] (v1);
  \draw (u1) edge [original edge] (w1);
  \draw (v1) edge [original edge] (w1);
  
  \node [junction] (u2) at (0+5.5,0) {$\uVertex$};
  \node [junction] (v2) at (-0.75+5.5,1.5) {$\vVertex$};
  \node [junction] (w2) at (1+5.5,2.5) {$\wVertex$};
  
  \draw[sloped,below] (u2) edge [original edge] node[edge label]{$\augLen(v,u)$} (v2);%<-
  \draw[sloped,below] (u2) edge [original edge] node[edge label]{$\augLen(u,w)$} (w2);%->
  \draw[sloped,above] (v2) edge [original edge] node[edge label]{$\augLen(v,w)$} (w2);%->
  \draw[edge,->,rounded corners=0.2cm,KITgreen] (v2.-43) -- (u2.92) -- (w2.-132);
  
  \node[junction] (u3) at (0+9,0) {$\uVertex$};
  \node[junction] (v3) at (-0.75+9,1.5) {$\vVertex$};
  \node[junction] (w3) at (1+9,2.5) {$\wVertex$};
  
  \draw[sloped,below] (u3) edge [original edge] node[edge label]{$\augLen(u,v)$} (v3);%->
  \draw[sloped,below] (u3) edge [original edge] node[edge label]{$\augLen(w,u)$} (w3);%<-
  \draw[sloped,above] (v3) edge [original edge] node[edge label]{$\augLen(w,v)$} (w3);%<-
  \draw[edge,<-,rounded corners=0.2cm,KITgreen] (v3.-43) -- (u3.92) -- (w3.-132);
  
  \draw[->, original edge] (-2,-0.4) -- node[anchor=south,sloped] {rank} (-2,2.9);
\end{tikzpicture}

%% file: Graphics/directedContraction.tex
\begin{tikzpicture}
  [junction/.style={fill=KITblack!5,draw, circle, inner sep=0pt, minimum size=0.66cm},
   point/.style={draw, circle, inner sep=0pt, minimum size=1ex},
   edge/.style={>=stealth, shorten >=1pt, shorten <=1pt},
   original edge/.style={edge, draw=black},]

  \node [junction] (u1) at (0,0) {$\uVertex$};
  \node [junction] (v1) at (1.6666,0.75) {$\vVertex$};
  \node [junction] (w1) at (0.3333,1.75) {$\wVertex$};
  \node [junction] (x1) at (2,2.5) {$x$};
  
  \draw (u1) edge [original edge,->] (v1);
  \draw (u1) edge [original edge,->] (w1);
  \draw (v1) edge [original edge,->] (x1);
  \draw (w1) edge [original edge,->] (x1);
  
  \node [junction] (u2) at (0+5.5,0) {$\uVertex$};
  \node [junction] (v2) at (1.6666+5.5,0.75) {$\vVertex$};
  \node [junction] (w2) at (0.3333+5.5,1.75) {$\wVertex$};
  \node [junction] (x2) at (2+5.5,2.5) {$x$};
  
  \draw (u2) edge [original edge, ultra thick] (v2);
  \draw (u2) edge [original edge] (w2);
  \draw (v2) edge [original edge] (x2);
  \draw (w2) edge [original edge, ultra thick] (x2);
  \draw (v2) edge [original edge, ultra thick] (w2);
  
%  \node [junction] (u2) at (0+5.5,0) {$\uVertex$};
%  \node [junction] (v2) at (-0.75+5.5,1.5) {$\vVertex$};
%  \node [junction] (w2) at (1+5.5,2.5) {$\wVertex$};
%  
%  \draw (u2) edge [original edge,->] (v2);
%  \draw (u2) edge [original edge,->] (w2);
%  \draw (v2) edge [original edge,->] (w2);
%  
%  \node [junction] (u3) at (0+9,0) {$\uVertex$};
%  \node [junction] (v3) at (-0.75+9,1.5) {$\vVertex$};
%  \node [junction] (w3) at (1+9,2.5) {$\wVertex$};
%  
%  \draw (u3) edge [original edge,<-] (v3);
%  \draw (u3) edge [original edge,<-] (w3);
%  \draw (v3) edge [original edge,<-] (w3);
  
  \draw[->, original edge] (-1.25,-0.4) -- node[anchor=south,sloped] {rank} (-1.25,2.9);
\end{tikzpicture}

%% file: Graphics/SeperatorDecomposition.tex
\begin{tikzpicture}[x=1cm,y=0.6cm,edge/.style={>=stealth},original edge/.style={edge, draw=black},]	
	\begin{scope}[xshift=12.5cm]
		\node[anchor=north west,inner sep=0pt,outer sep=0pt] at (-4,6) {BFS order:};
		\draw[fill,ultra thick] (-4,5.4) rectangle (0.5,4.8);
		\fill[KITgreen] (0.5,5.4) rectangle (0,4.8);
		\fill[KITblue] (0,5.4) rectangle (-0.5,4.8);
		\fill[KITred] (-0.5,5.4) rectangle (-1,4.8);
		\fill[KITblue30] (-1,5.4) rectangle (-1.4,4.8);
		\fill[KITred30] (-1.4,5.4) rectangle (-1.8,4.8);
		\fill[KITblue30] (-1.8,5.4) rectangle (-2.3,4.8);
		\fill[KITred30] (-2.3,5.4) rectangle (-2.6,4.8);
		\fill[KITblue30] (-2.6,5.4) rectangle (-2.9,4.8);
		\fill[KITred30] (-2.9,5.4) rectangle (-3.2,4.8);
		\fill[KITblue30] (-3.2,5.4) rectangle (-3.5,4.8);
		\fill[KITred30] (-3.5,5.4) rectangle (-3.65,4.8);
		\fill[KITblue30] (-3.65,5.4) rectangle (-3.8,4.8);	
		\fill[KITred30] (-3.8,5.4) rectangle (-3.9,4.8);
		\fill[KITblue30] (-3.9,5.4) rectangle (-4,4.8);	
		
		\node[anchor=north west,inner sep=0pt,outer sep=0pt] at (-4,4) {DFS post-order:};
		\draw[fill,ultra thick] (-4,3.4) rectangle (0.5,2.8);
		\fill[KITgreen] (0.5,3.4) rectangle (0,2.8);
		\fill[KITblue] (0,3.4) rectangle (-0.5,2.8);
		\fill[KITblue30] (-0.5,3.4) rectangle (-2.0,2.8);
		\fill[KITred] (-2.0,3.4) rectangle (-2.5,2.8);
		\fill[KITred30] (-2.5,3.4) rectangle (-4,2.8);

		\draw[->, original edge] (-4,1.75) -- node[anchor=north,sloped] {rank} (0.5,1.75);
	\end{scope}

	\newcommand{\contour}{(2,3.5) (1,2) (3,1) (4,2) (6,1) (7,2) (6,3) (7,5) (5,6) (4,5) (2,6) (1,5)}
	
	\begin{scope}	
		\clip plot[smooth cycle, tension=0.75] coordinates{\contour};
		
		\fill[KITblue30] (0,0) rectangle (4,7);
		\draw[decorate, decoration={snake, segment length=2cm,amplitude=3pt}, line width=0.2cm,KITblue] (4,3.3) -- (1.5,3.5);
		
		\fill[KITred30] (8,0) rectangle (4,7);	
		\draw[decorate, decoration={snake, segment length=2cm,amplitude=3pt}, line width=0.2cm,KITred] (4,4.1) -- (7,2.6);
		
		\draw[decorate, decoration={snake, segment length=2cm,amplitude=3pt}, line width=0.35cm,KITgreen] (4,6) -- (4,1);
	\end{scope}

	\draw[thick] plot[smooth cycle, tension=0.75] coordinates{\contour};
\end{tikzpicture}

%% file: table/preprocessing.tex
\begin{tabular}{lrrrrr}
\toprule
{} &  Ordering & \multicolumn{3}{c}{Contraction} & Total \\ \cmidrule(lr){3-5}
{} &  {} &  Chordal &  Contraction & Dyn. adjacency & \\
{} &  {} &  completion & graph \cite{DBLP:journals/jea/DibbeltSW16} & array \cite{DBLP:journals/jea/DibbeltSW16} & \\
\midrule
Stuttgart &                          0.8 &                          0.0 &          -- &           -- &             0.9 \\
Germany   &                        203.9 &                          1.3 &          -- &           -- &           222.9 \\
Europe    &                        341.2 &                          1.6 &        15.5 &        305.8 &           361.1 \\
\bottomrule
\end{tabular}

%% file: table/customization.tex
\begin{tabular}{llrrrrrrrr}
\toprule
& & \multicolumn{4}{c}{Stuttgart [ms]} & \multicolumn{4}{c}{Europe [s]} \\ \cmidrule(l{3pt}r{3pt}){3-6} \cmidrule(l{3pt}r{3pt}){7-10}
Impl & Threads &     Basic & Perfect & Construct &  Total &  Basic & Perfect & Construct &  Total \\
\midrule
\cite{DBLP:journals/jea/DibbeltSW16}
& 1           &           &         &           &        &  10.88 &   22.02 & $\approx$ 9.39 & $\approx$ 42.35 \\
& 16          &           &         &           &        & \multicolumn{2}{c}{5.47} & $\approx$ 9.39 & $\approx$ 14.86 \\
\addlinespace
\cite{DBLP:journals/jea/Buchhold0W19}
& 1           &     20.51 &   20.77 &     48.64 &  89.93 &   5.60 &    6.48 &      9.39 &  21.47 \\
& 16          &      4.91 &    4.41 &      4.35 &  13.66 &   1.11 &    0.63 &      0.80 &   2.54 \\
\addlinespace
{[ours]}
& 1           &     19.85 &   18.95 &     10.33 &  49.13 &   4.09 &    4.72 &      1.95 &  10.76 \\
& 2           &     11.34 &    9.90 &      5.69 &  26.93 &   2.05 &    2.42 &      0.99 &   5.46 \\
& 4           &      7.39 &    5.65 &      3.57 &  16.61 &   1.22 &    1.26 &      0.54 &   3.02 \\
& 8           &      5.70 &    3.51 &      2.46 &  11.67 &   0.81 &    0.69 &      0.35 &   1.86 \\
& 16          &      6.25 &    3.10 &      2.59 &  11.94 &   0.58 &    0.37 &      0.30 &   1.25 \\
\bottomrule
\end{tabular}

%% file: table/queries.tex
\begin{tabular}{llrrrrrrrr}
\toprule
& {}      & \multicolumn{3}{c}{Search space}     & \multicolumn{4}{c}{Running time [$\mu$s]}               & \# Path    \\ \cmidrule(l{3pt}r{3pt}){3-5} \cmidrule(l{3pt}r{3pt}){6-9}
& {} & \multirow{2}{*}{Vertices} & \multicolumn{2}{c}{Edges} & \multicolumn{2}{c}{Distance} & \multicolumn{2}{c}{Path} & Vertices \\ \cmidrule(l{3pt}r{3pt}){4-5} \cmidrule(l{3pt}r{3pt}){6-7} \cmidrule(l{3pt}r{3pt}){8-9}
& & &                  \augGraph &      \minAugGraph &           \augGraph &   \minAugGraph &     \augGraph &   \minAugGraph &         \\
\midrule
Stuttgart & Travel time &                    216.4 &            8\,846.5 &    3\,837.1 &            22.9 &   14.5 &       5.6 &    5.4 &             185.9 \\
       & Geo distance &                    216.4 &            8\,888.6 &    3\,303.3 &            22.2 &   13.0 &       4.4 &    4.2 &             149.8 \\
\addlinespace
Germany & Travel time &                   1\,277.9 &          274\,113.1 &   78\,474.2 &           442.0 &  163.2 &     234.6 &  134.4 &            4\,681.0 \\
       & Heavy traffic &                   1\,277.9 &          274\,478.3 &   85\,298.9 &           435.0 &  174.2 &     241.7 &  179.5 &            5\,363.4 \\
       & Geo distance &                   1\,277.9 &          274\,788.5 &  131\,985.4 &           438.3 &  246.3 &     383.4 &  336.2 &            6\,174.7 \\
\addlinespace
Europe & Travel time &                   1\,041.2 &          186\,006.5 &   69\,312.9 &           300.2 &  137.5 &      95.9 &   71.3 &            1\,389.8 \\
       & Geo distance &                   1\,041.2 &          185\,701.3 &   92\,616.8 &           303.3 &  177.9 &     281.3 &  252.9 &            3\,158.9 \\
\bottomrule
\end{tabular}